\documentclass[journal]{IEEEtran}
\usepackage[utf8]{inputenc} 
\usepackage[T1]{fontenc}    
\usepackage{hyperref}       
\usepackage{booktabs}       
\usepackage{amsfonts}       
\usepackage{nicefrac}       
\usepackage{microtype}      
\usepackage{xcolor}         
\usepackage{graphicx}
\usepackage{amsthm}
\usepackage{amsmath}
\usepackage{amssymb}
\usepackage{siunitx}
\usepackage{stfloats}
\usepackage{subfigure}
\usepackage{hyperref}
\usepackage{enumerate}
\usepackage{ulem}
\usepackage{enumitem}
\usepackage{xcolor}

\usepackage{ulem}
\usepackage{bm}
\usepackage{nicefrac}
\usepackage{booktabs}
\usepackage{array}
\usepackage{multirow}
\usepackage{threeparttable}
\usepackage{makecell}
\usepackage{caption}

\usepackage{amsmath}

\usepackage{amssymb}
\usepackage{bm}
\usepackage{nicefrac}
\usepackage{booktabs}
\usepackage{array}
\usepackage{multirow}
\usepackage{threeparttable}
\usepackage{makecell}
\usepackage[procnumbered,ruled,vlined,linesnumbered]{algorithm2e}
\usepackage{stfloats}
\usepackage{graphicx}
\usepackage{subfigure}
\usepackage{enumerate}
\usepackage{epstopdf}
\usepackage{tabularx}
\usepackage{stfloats}
\usepackage{booktabs}
\usepackage{siunitx}
\usepackage{hyperref}
\usepackage{balance}

\newtheorem{problem}{Problem}
\newtheorem{theorem}{Theorem}[section]

\newtheorem{lemma}[theorem]{Lemma}

\newtheorem{fact}[theorem]{Fact}

\newtheorem{remark}{Remark}

\newtheorem{definition}[theorem]{Definition}

\def\calG{\mathcal{G}}

\newcommand{\removelatexerror}{\let\@latex@error\@gobble}

\def\calG{\mathcal{G}}
\def\calH{\mathcal{H}}

\def\norm#1{\left\| #1 \right\|}
\def\len#1{\left\lVert #1 \right\rVert}
\def\kh#1{\left( #1 \right)}

\def\len#1{\left| #1 \right|}

\newcommand\LL{\bm{\mathit{L}}}

\newcommand\KK{\boldsymbol{\mathit{K}}}
\newcommand\XX{\boldsymbol{\mathit{X}}}
\newcommand\yy{\boldsymbol{\mathit{y}}}
\newcommand\zz{\boldsymbol{\mathit{z}}}
\newcommand\xx{\boldsymbol{\mathit{x}}}
\newcommand\ff{\boldsymbol{\mathit{f}}}
\newcommand\aaa{\boldsymbol{\mathit{a}}}

\newcommand\bb{\boldsymbol{\mathit{b}}}
\newcommand\cc{\boldsymbol{\mathit{c}}}

\newcommand\ee{\boldsymbol{\mathit{e}}}

\newcommand\qq{\boldsymbol{\mathit{q}}}

\newcommand\uu{\boldsymbol{\mathit{u}}}
\newcommand\sss{\boldsymbol{\mathit{s}}}
\newcommand\hh{\boldsymbol{\mathit{h}}}

\renewcommand\SS{\boldsymbol{\mathit{S}}}
\renewcommand\AA{\boldsymbol{\mathit{A}}}
\newcommand\BB{\boldsymbol{\mathit{B}}}

\newcommand\DD{\boldsymbol{\mathit{D}}}

\newcommand\PP{\boldsymbol{\mathit{P}}}
\newcommand\MM{\boldsymbol{\mathit{M}}}
\newcommand\TT{\boldsymbol{\mathit{T}}}
\newcommand\YY{\boldsymbol{\mathit{Y}}}

\newcommand\QQ{\boldsymbol{\mathit{Q}}}

\newcommand\II{\boldsymbol{\mathit{I}}}
\newcommand\OO{\boldsymbol{\mathit{O}}}
\newcommand\UU{\boldsymbol{\mathit{U}}}
\newcommand\VV{\boldsymbol{\mathit{V}}}
\newcommand\vvv{\boldsymbol{\mathit{v}}}
\newcommand{\SDDMSolver}{\textsc{Solver}}

\DeclareMathOperator*{\argmax}{arg\,max}

\DontPrintSemicolon
\SetKw{KwAnd}{and}
\SetFuncSty{textsc}
\SetKwInOut{Input}{Input\ \ \ \ }

\SetKwInOut{Output}{Output}

\usepackage{tabularx}
\usepackage{stfloats}

\DontPrintSemicolon
\SetKw{KwAnd}{and}
\SetFuncSty{textsc}
\SetKwInOut{Input}{Input\ \ \ \ }
\SetKwInOut{Output}{Output}

\usepackage{tabularx}
\usepackage{stfloats}

\ifCLASSOPTIONcompsoc
  \usepackage[nocompress]{cite}
\else
  \usepackage{cite}
\fi

\ifCLASSINFOpdf

\else

\fi

\begin{document}
\title{Friedkin-Johnsen Model for Opinion Dynamics on Signed Graphs}

\author{Xiaotian~Zhou, 
        Haoxin~Sun,
        Wanyue~Xu,
        Wei~Li,
        and~Zhongzhi~Zhang,~\IEEEmembership{Member,~IEEE}

\IEEEcompsocitemizethanks{
\IEEEcompsocthanksitem This work was supported by the National Natural Science Foundation of China (Nos. U20B2051, 62372112, and 61872093). \textit{(Corresponding author: Zhongzhi~Zhang.)} 

\IEEEcompsocthanksitem  Xiaotian~Zhou, Haoxin~Sun, Wanyue~Xu, and Zhongzhi Zhang  are with Shanghai Key Laboratory of Intelligent Information Processing, School of Computer Science, Fudan University, Shanghai 200433, China. Wei~Li is with the Academy for Engineering and Technology, Fudan University, Shanghai, 200433, China.  
\protect\\
E-mail: 22110240080@m.fudan.edu.cn, 21210240097@m.fudan.edu.cn, xuwy@fudan.edu.cn, fd\_liwei@fudan.edu.cn, zhangzz@fudan.edu.cn
}
\thanks{Manuscript received xxxx; revised xxxx.}
}

\markboth{IEEE Transactions on Knowledge and Data Engineering,~Vol.~xx, No.~xx, December~2023}%
{Shell \MakeLowercase{\textit{et al.}}: Bare Demo of IEEEtran.cls for Computer Society Journals}

\maketitle
\begin{abstract}
A signed graph offers richer information than an unsigned graph, since it  describes both collaborative and competitive relationships in social networks. In this paper, we study opinion dynamics on a signed graph, based on the Friedkin-Johnsen model. We first interpret the equilibrium opinion in terms of a defined random walk on an augmented signed graph, by representing the equilibrium opinion of every node as a combination of all nodes' internal opinions, with the coefficient of the internal opinion for each node being the difference of two absorbing probabilities. We then quantify some relevant social phenomena and express them in terms of the $\ell_2$ norms of vectors. We also design a nearly-linear time signed Laplacian solver for assessing these quantities, by establishing a connection between the absorbing probability of random walks on a signed graph and that on an associated unsigned graph. We further study the opinion optimization problem by changing the initial opinions of a fixed number of nodes, which can be optimally solved in cubic time. We  provide a nearly-linear time algorithm with error guarantee to approximately solve the problem. Finally, we execute extensive experiments on sixteen real-life signed networks, which show that both of our algorithms are effective and efficient, and are scalable to massive graphs with over 20 million nodes.
\end{abstract}



\IEEEpeerreviewmaketitle


\section{Introduction}

\IEEEPARstart{D} {} ue to the ever-increasing availability of the power of computing, storage, and manipulation, in the past years, social media and online social networks have experienced explosive growth~\cite{Le20}, which have constituted an important part of people's lives~\cite{SmCh08}, leading to a substantial change of the ways that people exchange and form opinions~\cite{JiMiFrBu15,AnYe19} 
regarding voter, product marketing, social hotspots, and so on. Numerous recent studies have shown that online social networks and social media play a vital role during the whole process of opinion dynamics, including opinion diffusion, evolution, as well as formation~\cite{DaGoMu14,FoPaSk16,AuFeGr18}. In view of the practical relevance, opinion dynamics in social networks has received considerable recent attention from the scientific community, spanning various aspects of this dynamical process, such as modelling opinion formation, quantifying some resultant social phenomena (e.g., polarization~\cite{MaTeTs17,MuMuTs18}, disagreement~\cite{MuMuTs18}, and conflict~\cite{ChLiDe18}), and optimizing opinion~\cite{ GiTeTs13,AbKlPaTs18,ChLiSo19,MaMiTaTa21}.

The most important step in the study of opinion dynamics is probably the establishment of a mathematical model. Over the past years, a rich variety of models have been proposed for modelling opinion dynamics. Among these models,  the Friedkin-Johnsen (FJ) model~\cite{FrJo90} is a popular one, which has found practical applications. For example, the concatenated FJ model has been applied to explain the Paris Agreement negotiation process~\cite{BeWaVaHoShAl21} and the FJ model with multidimensional opinions has been used to account for the change of public opinion on the US-Iraq War~\cite{FrPrTePa16}. To the best of our knowledge, the FJ model is the only opinion dynamics model that has been confirmed by a sustained line of human-subject experiments~\cite{FrPrTePa16,FrBu17,BeWaVaHoShAl21}.
Thus, the FJ model, despite being simple, has received much recent interest. It has been recently used to quantify various social phenomena such as disagreement~\cite{MuMuTs18,DaGoLe13}, polarization~\cite{DaGoLe13,MaTeTs17,MuMuTs18}, and conflict~\cite{ChLiDe18}, for which some nearly-linear time approximation algorithms were designed~\cite{XuBaZh21,XuZhGuZhZh22}. Moreover, optimization of the overall opinion for the FJ model was also heavily studied~\cite{AhDeHaMaYa15,GiTeTs13,AbKlPaTs18,ChLiSo19,AbChKlLiPaSoTs21,MaMiTaTa21,SuZh23}. Finally, diverse variants of this well-known model have been introduced~\cite{JiMiFrBu15,AnYe19}, by considering different factors affecting opinion formation, such as susceptibility to persuasion~\cite{AbKlPaTs18}, peer pressure~\cite{SeGrSqRa19}, stubbornness~\cite{XuZhGuZhZh22}, and algorithmic filtering~\cite{ChMu20}.

Because of the  outstanding merits of the FJ model detailed above, in this work we also choose this model as our subject.
Note that in addition to the aforementioned aspects, opinion shaping is also significantly affected by the interactions among individuals. Most previous studies for the FJ model only capture positive or cooperative interactions among agents described by an unsigned graph, neglecting those negative or competitive interactions, in spite of the fact that both friendly and hostile relationships often exist simultaneously in many realistic scenarios, especially in social networks~\cite{ShAlBa19,ShPrJoBaJo16}.  In view of the ubiquity of competitive interactions in real systems, the FJ model on signed graphs has been built and studied~\cite{XuHuWu20,RaHo21,HeZhLiRu20,HeZeZhLi22}, which takes into account both collaborative and antagonistic relationships, providing a comprehensive understanding of human relationships in opinion dynamics~\cite{TaChAgLi16}.

However, the properties and complexity of various graph
problems change once there are negative edges. For instance, for a signed graph having cycles with negative edges, the shortest-path problem is known to be NP-hard~\cite{ChNaTeDh11}. In the context of the signed FJ model, the inclusion of antagonistic relationships presents more challenges to analyze relevant properties and develop good algorithms. First, the probability explanation for the equilibrium opinions in unsigned graphs~\cite{GiTeTs13}  no longer holds for their signed counterparts. Moreover, measures and their expressions for some social phenomena (polarization, disagreement, and conflict) are not fully defined and studied. Finally, existing fast approaches for eventuating or optimizing  relevant quantifies about opinion optimization on unsigned graphs~\cite{XuBaZh21,XuZhGuZhZh22} cannot be directly carried over to signed graphs, while the prior algorithm for signed graphs has cubic complexity~\cite{XuHuWu20} and is thus not applicable to large signed graphs. These enlighten us to solve the challenges for the FJ model on signed graphs.

%


In order to fill the aforementioned gap, in this paper, we present an extensive study for the FJ model on a signed graph.
Our main  contributions  are summarized as follows.
\begin{itemize}[leftmargin=10pt,topsep=6pt]

\item We interpret the equilibrium expressed opinion in terms of a newly defined random walk on an augmented signed graph. We represent the expressed opinion of each node as a combination of the internal opinions of all nodes, where the coefficient of every node's internal opinion is the difference of two absorbing probabilities.

\item  We provide measures and expressions for some relevant social phenomena, including conflict, polarization, and disagreement based on signed FJ model. Particularly, we express each quantity in terms of the $\ell_2$ norm of a vector.

\item We construct an unsigned graph for any signed graph, and prove the equivalence of the absorbing probabilities for random walks on the augmented signed graph and its associated augmented unsigned graph. Utilizing this connection, we develop a nearly-linear time signed Laplacian solver approximating the equilibrium expressed opinion and relevant social phenomena.

\item We address the problem of maximizing (or minimizing) the overall opinion by selecting a small fixed number of nodes and changing their initial opinions, which can be optimally solved in cubic time. We  propose an approximation algorithm solving the problem, which has nearly-linear time complexity and an error guarantee, based on our proposed signed Laplacian solver.

\item We demonstrate the performance of our two efficient algorithms by performing extensive experiments on realistic signed graphs, which indicate that both algorithms are fast and accurate, and are scalable to graphs with over 20 million nodes.
\end{itemize}

\section{Preliminaries}
In this section, we provide an overview of the notations, signed graphs and their related matrices, and the FJ model.

\subsection{Notations}

For a vector $\aaa$, the $i$-th element is denoted as $\aaa_i$. For a matrix $\AA$, the element at the $i$-th row and $j$-th column is denoted as $\AA_{i,j}$; the $i$-th row and $j$-th column are represented as $\AA_{i,:}$ and $\AA_{:,j}$, respectively. For matrix $\AA$ and vector $\aaa$, $\AA^{\top}$ and $\aaa^{\top}$ denote, respectively, their transpose. $\II$ and $\OO$ denote, respectively, the identity matrix and the zero matrix of the appropriate dimension. The vector $\ee_i$ is a vector of appropriate dimension, where the $i$-th element is 1 and all other elements are 0’s. The vector $\mathbf{0}$ (or $\mathbf{1}$)  is a vector of appropriate dimension with all entries equal to 0 (or 1). For a vector $\aaa$, its $\ell_2$ norm  is $\norm{\aaa}_2 = \sqrt{\sum_i \aaa_i^2}$, its $\ell_0$ norm is denoted as $\norm{\aaa}_0$ defined as the number of nonzero elements in $\aaa$, and its norm with respect to a matrix $\AA$ is $\norm{\aaa}_{\AA} = \sqrt{\aaa^\top \AA \aaa}$.

For two nonnegative scalars $a$ and $ b$ and $0 < \epsilon < 1/2$, we say that $a$ is an $\epsilon$-approximation of $b$, denoted by $a \approx_{\epsilon} b$, if $(1-\epsilon) a \leq b \leq (1+\epsilon) a$. For two positive semidefinite matrices $\XX$ and $\YY$, we say that $\XX \preceq \YY$ if $\YY - \XX$ is positive semidefinite, meaning that $\xx^\top \XX \xx \leq \xx^\top \YY \xx$ holds for any real vector $\xx$.

\subsection{Signed Graph and Their Related Matrices}

Let $\calG= (V,E,w)$ denote a connected undirected signed graph with $n=|V|$ nodes, $m=|E|$ edges, and edge sign function $w : E \to \{-1,+1\}$. We call an edge $e=(i,j)\in E$ a positive (or negative) edge if $w(i,j)$ is +1 (or -1). The edge sign $w(i,j)$ of edge $e=(i,j)\in E$ represents the relationship between node $i$ and node $j$. If $w(i,j)=+1$, nodes $i$ and  $j$ are cooperative; and if $w(i,j)=-1$, nodes $i$ and  $j$ are competitive. Let $N_i$ be the set of neighbours of node $i$, which can be classified into two disjoint subsets: the friend set $N^F_i$ and the enemy (or opponent) set $N^E_i$. Any node in $N^F_i$ is directly connected to $i$ by a positive edge, while any node in $N^E_i$ is linked to $i$ by a negative edge. The degree $d_i$ of node $i$ is defined as $d_i=\sum_{j\in N_i} |w(i,j)|$.

For an undirected signed graph $\calG$, let $\DD $ be its degree matrix, which is a diagonal matrix defined as $\DD=\text{diag}\{d_1,d_2,\ldots,d_n\}$. Let $\AA \in \mathcal{R}^{n \times n}$ be its signed adjacency matrix, with the entry $\AA_{i,j}$ defined as follows: $\AA_{i,j} = w(i,j)$ for any edge $(i,j)\in E$, and $\AA_{i,j} =0$  otherwise. The signed Laplacian matrix $\LL$ of $\calG$ is defined as $\LL=\DD-\AA$. For any edge in $\calG$, we fix an arbitrary orientation. Then we can define the edge-node incidence matrix $\BB_{m\times n}$ of $\calG$, whose entries are defined as follows. For any edge $e=(u,v)$, $\BB_{e,u}=1$ if node $u$ is the head of edge $e$; $\BB_{e,v}=-w(u,v)$ if node $v$ is the tail of $e$; and $\BB_{e,x}=0$ for $x \neq u,v$. Then, matrix $\LL$ can be written as $\LL=\BB^\top \BB$, implying that $\LL$ is symmetric and positive semidefinite. Note that matrix $\LL$ is called ``opposing'' Laplacian in~\cite{ShAlBa19}, where another Laplacain matrix named ``repelling'' Laplacian is also defined. In this paper, we focus on the  ``opposing'' Laplacian $\LL$, called Laplacian matrix hereafter.

According to the sign of edges, any signed graph $\calG$ can be divided into two spanning subgraphs: the positive graph $\calG^+=(V,E^+,w^+)$ and the negative graph $\calG^-=(V,E^-,w^-)$. Both $\calG^+$ and $\calG^-$ share the same node set $V$ as graph $\calG$.  $\calG^+$ contains all positive edges in $E$, while $\calG^-$ includes all negative edges in $E$. In other words, $E=E^+ \cup E^-$, and for any edge $e \in E $ with end nodes $i$ and $j$, $w^+(i,j)=1$ if $e \in E^+$,  $w^-(i,j)=-1$ if $e \in E^-$, and  $w^+(i,j)= w^-(i,j)=0$ otherwise. Let $\AA^+$ and $\AA^-$ be the adjacency matrices of $\calG^+$ and $\calG^-$, respectively. Then, $\AA^+$ is a non-negative matrix, while $\AA^-$ is a non-positive matrix. Similarly, for graphs $\calG^+$ and $\calG^-$, we use $\DD^+$ and $\DD^-$ to denote their degree matrices, use $\BB^+$ and $\BB^-$ to represent their incidence matrices, and use $\LL^+$ and $\LL^-$ to denote their Laplacian matrices.

We now provide some useful matrix relations involving the signed Laplacian matrix $\LL$ in Fact~\ref{fact}, which will be used in the proofs in the sequel.
\begin{fact}~\label{fact}
    We have $\LL^+ \preceq \LL$, $\LL^- \preceq \LL$, $\II \preceq \II + \LL$,  $\LL \preceq \II + \LL$, $ \II + \LL \preceq 2n \II$, and $ \frac{1}{2n}\LL \preceq \II$.
\end{fact}

\subsection{Friedkin-Johnsen Model on Signed Graphs}



The Friedkin-Johnsen (FJ) model is a popular model for analysing opinion evolution and formation on graphs. In the FJ model~\cite{FrJo90}, each node or agent $i\in V$ is associated with two opinions. One is the internal opinion $\sss_i$, which is a constant value in the interval $[-1,1]$, reflecting the intrinsic position of node $i$ on a certain topic. The other is the expressed opinion $\zz_i(t)$ at time $t$. At each time step, node $i$ updates its expressed opinion according to the rule of minimizing its psycho-social cost $(\zz_i(t)-\sss_i)^2+\sum_{j\in N_i}(\zz_i(t)-\AA_{i,j}\zz_j(t))^2$, which is a function of its expressed opinion, its neighbours' expressed opinions, and its internal opinion~\cite{BiKlOr15,HeZhLiRu20,RaHo21}. To minimize this social cost function, each agent updates its expressed opinion as follows
\begin{equation}\label{FJ}
	\zz_i(t+1) = \frac{\sss_i +\sum_{j\in N_i}\AA_{i,j}\zz_j(t)}{1+d_i}.
\end{equation}
That is to say, the expressed opinion $\zz_i(t+1)$ for node $i$ at time $t+1$ is updated by  averaging its internal opinion $\sss_i$, the expressed opinions of its friends at time $t$, and the opposite expressed opinions of its opponents at time $t$. We define the initial opinion vector as $\sss = (\sss_1,\sss_2,\ldots,\sss_n)^\top$, and define the vector of expressed opinions at time $t$ as $\zz(t) = (\zz_1(t),\zz_2(t),\ldots,\zz_n(t))^\top$. After long-time evolution, the expressed opinion vector converges to an equilibrium vector $\zz = (\zz_1,\zz_2,\ldots,\zz_n)^\top$ satisfying
\begin{equation}\label{equ:exp}
    \zz= \lim_{t \to \infty} \zz(t)=(\II+\LL)^{-1}\sss.
\end{equation}

Although on both signed and unsigned graphs, the dynamics equation~\eqref{FJ} and the equilibrium opinion~\eqref{equ:exp} for the FJ model have the same form of mathematical expressions. There are some notable differences between the unsigned FJ model~\cite{FrJo90} and the signed FJ model~\cite{XuHuWu20,HeZhLiRu20}. For example, their fundamental matrix~\cite{GiTeTs13} $(\II+\LL)^{-1}$ have obviously different properties. For an unsigned graph, matrix $(\II+\LL)^{-1}$ is a doubly stochastic matrix, with every element being positive~\cite{GiTeTs13,MaTeTs17}. As a result, the expressed opinion of each node is a convex combination of the internal opinions of all nodes. In contrast, for a signed graph, $(\II+\LL)^{-1}$ is generally not a positive matrix and is thus not doubly stochastic. Thus, the expressed opinion of each node is not a convex combination of the internal opinions of all nodes. This poses an important challenge for computing expressed opinion vector and analyzing and optimizing relevant problems for opinion dynamics on signed graphs since existing algorithms for unsigned graphs are no longer applicable.

\section{Interpretation of Expressed Opinion for the Signed FJ Model}

The unsigned FJ model has been interpreted from various perspectives. For example, the opinion evolution process has been explained according to Nash equilibrium~\cite{BiKlOr15}, and the equilibrium opinion has been accounted for based on absorbing random walks~\cite{GiTeTs13} or spanning diverging forests~\cite{XuBaZh21,XuZhGuZhZh22,SuZh23}. However, these interpretations do not hold for signed graphs anymore, since those methods are based on the unsigned graph structure and do not apply to signed graphs.

In this section, we give an explanation for the equilibrium expressed opinions of the signed FJ model. For this purpose, we introduce absorbing random walks on signed graphs, based on which we provide an interpretation of the equilibrium expressed opinion of each node in terms of the absorbing probabilities, and highlight the parallels and distinctions between the FJ model on signed and unsigned graphs.


We begin by extending the signed graph $\calG=(V,E,w)$ to an augmented graph $\calH=(X,R,y)$ with absorbing states, defined as follows:
\begin{enumerate}[leftmargin=12pt,topsep=6pt]
    \item The node set $X$ is defined as $X=V \cup \bar{V}$, where $\bar{V}$ is a set of $n$ nodes such that for each node $i\in V$, there is a copy node $\sigma(i) \in \bar{V}$ of node $i$;
    \item The edge set $R$ includes all the edges $E$ of $\calG$, plus a new set of edges between each node $i\in V$ and the copy node $\sigma(i) \in \bar{V}$ of node $i$. That is, $R=E \cup \bar{E}$, where $\bar{E}=\{(i,\sigma(i))| i\in V \}$;
    \item The edge sign $y(e=(i,\sigma(i)))$ of each new edge $e=(i,\sigma(i)) \in \bar{E}$ is set to +1. For edge $e=(i,j) \in E$, its edge sign $y(e=(i,j))$ is identical to the sign $w(e=(i,j))$ of the corresponding edge in $\calG$.
\end{enumerate}

Neglecting the sign of each edge, we can define an absorbing random walk on this augmented graph $\calH$, where nodes in $V$ are transient states and nodes in $\bar{V}$ are absorbing states. Then,
the transition matrix $\PP$ for this  absorbing random walk has the form
\begin{equation*}
    \PP=\begin{bmatrix}
        (\II+\DD)^{-1} (\AA^+-\AA^-) & (\II+\DD)^{-1}\\
         \OO & \II
    \end{bmatrix}
\end{equation*}



For a random walk $r=(u_0,u_1,\ldots,u_t)$ on signed graphs, we introduce a sign $ l(r)$ as~\cite{LiChWaZh13,NePe22}
\begin{equation*}
   l(r)=\text{sign}\left(\prod_{x=1}^t w(u_{x-1},u_{x})\right),
\end{equation*}
where $\text{sign}(x)$ denotes the sign of nonzero $x$ defined as: $\text{sign}(x)=+1$ if $x>0$, $\text{sign}(x)=-1$ otherwise. Thus, the sign $l(r)$ of the walk $r$ is $+1$ if the number of its negative edges is even, and $ l(r)=- 1$ otherwise. For  a random walk $r$ on signed graphs, if $ l(r)=+1$ we call it a positive random walk; if $ l(r)=-1$ we call it a negative random walk.


Based on the above-defined two types of random walks, we define two absorbing probabilities for an absorbing random walk on a signed graph. For each node $i\in V$ and any absorbing state $\sigma(j)\in \bar{V}$, we define two probabilities as follows:
\begin{itemize}[leftmargin=12pt,topsep=6pt]
\item $p_{i,\sigma(j)}$: the probability that a positive random walk starting from transient state $i$ is absorbed by node $\sigma(j)$.
\item $q_{i,\sigma(j)}$: the probability that a negative random walk starting from transient state $i$ is absorbed by node $\sigma(j)$.
\end{itemize}
Using these two absorbing probabilities, we provide an interpretation of the equilibrium opinion for each node.

\begin{theorem}
\label{thm:randomwalk}
  For the FJ model of opinion dynamics on a signed graph $\calG=(V,E,w)$, let $\sss$ be the initial opinion vector. Then the equilibrium expressed opinion of node $i\in V$ is expressed by
    \begin{equation*}
        \zz_i=\sum_{j\in V}(p_{i,\sigma(j)}-q_{i,\sigma(j)})\sss_j.
    \end{equation*}
\end{theorem}
\begin{proof}
First, we use a recursive method to derive the expressions for the two absorbing probabilities $p_{i,\sigma(j)}$ and $q_{i,\sigma(j)}$ for any $i\in V$ and $\sigma(j)\in \bar{V}$. To this end, for an absorbing random walk on graph $\calH$, we define two matrices, $\XX(t)$ and $\YY(t)$ for $t\geq 0$. The element $\XX(t)_{i,j}$ of $\XX(t)$ represents the probability that a positive random walk starting from a transient state $i$ arrives at a transient state $j$ at the $t$-th step. The element $\YY(t)_{i,j}$ of $\YY(t)$ denotes the probability that a negative random walk starting from a transient state $i$ reaches a transient state $j$ at the $t$-th step.

By definition of the signed random walk, we obtain the following recursive relations
\begin{align*}
&    \XX(t+1)_{i,:}=\XX(t)_{i,:}(\II+\DD)^{-1}\AA^+ -\YY(t)_{i,:}(\II+\DD)^{-1}\AA^- ,\\
 &\YY(t+1)_{i,:}=\YY(t)_{i,:}(\II+\DD)^{-1}\AA^+ -\XX(t)_{i,:}(\II+\DD)^{-1}\AA^-.
\end{align*}
Considering $\XX(0)=\II$ and $\YY(0)=\OO$, we obtain
\begin{small}
\begin{align*}
   &\XX(t)_{i,:}+\YY(t)_{i,:} = \ee_i^\top\left((\II+\DD)^{-1}(\AA^+-\AA^-)\right)^t ,\\
   &\XX(t)_{i,:}-\YY(t)_{i,:} = \ee_i^\top\left((\II+\DD)^{-1}(\AA^++\AA^-)\right)^t.
\end{align*}
Solving the above two equations leads to
\begin{align*}
   \XX(t)_{i,:} = &\frac{\ee_i^\top}{2}(((\II+\DD)^{-1}(\AA^+-\AA^-))^t\\
   &+((\II+\DD)^{-1}(\AA^++\AA^-))^t),\\
   \YY(t)_{i,:} = &\frac{\ee_i^\top}{2}(((\II+\DD)^{-1}(\AA^+-\AA^-))^t\\
   &-((\II+\DD)^{-1}(\AA^++\AA^-))^t).
\end{align*}
\end{small}

Note that if a walk arrives at $j\in V$ at a certain step, it will move to $\sigma(j)$ with probability $1/(1+d_j)$. Then the absorbing probability $p_{i,\sigma(j)}$ is exactly the sum of the probabilities of a positive random walk starting at node $i$ and being absorbed by node $\sigma(j)$ in $t=1,2,\ldots,\infty$ steps. Thus, we have
\begin{small}
\begin{align*}
    &\quad p_{i,\sigma(j)}= \sum_{t=0}^\infty \frac{1}{1+d_j} \XX(t)_{i,j}\\
    &=\frac{1}{2}\ee_i^\top ((\II+\DD-\AA^++\AA^-)^{-1} +(\II+\DD-\AA^+-\AA^-)^{-1})\ee_j.
\end{align*}
\end{small}
In a similar way, we obtain
\begin{small}
\begin{align*}
   &\quad  q_{i,\sigma(j)} = \sum_{t=0}^\infty \frac{1}{1+d_j} \YY(t)_{i,j}\\
    &=\frac{1}{2}\ee_i^\top ((\II+\DD-\AA^++\AA^-)^{-1}-(\II+\DD-\AA^+-\AA^-)^{-1})\ee_j.
\end{align*}
\end{small}
Using the obtained expressions for $p_{i,\sigma(j)}$ and $q_{i,\sigma(j)}$, we obtain
\begin{align*}
    &\sum_{j\in V} (p_{i,\sigma(j)} - q_{i,\sigma(j)}) \sss_j= \ee^\top_i(\II+\LL)^{-1}\sss=\zz_i,
\end{align*}
which completes the proof.
\end{proof}

Theorem~\ref{thm:randomwalk} shows that for each $i \in V$, its equilibrium expressed opinion $\zz_i$ is a weighted average of the internal  opinions of all nodes, with the weight of the internal opinion $s_j$ being $p_{i,\sigma(j)}- q_{i,\sigma(j)}$, where $j=1,2\ldots,n$. When there are no negative edges, the absorbing probability by negative random walks of each node is 0, and Theorem~\ref{thm:randomwalk} is consistent with the previous result obtained for unsigned graphs~\cite{GiTeTs13}. However, $p_{i,\sigma(j)}- q_{i,\sigma(j)}$ may be negative in the presence of negative edges, which is in sharp contrast to that for unsigned graphs, where the absorbing probability $p_{i,\sigma(j)}$ is always positive and $q_{i,\sigma(j)}$ is zero.


\section{Measures and Fast Evaluation of Some Social Phenomena}

Besides the equilibrium expressed opinions, quantifying some social phenomena based on the FJ model has also attracted much interest. Although some social phenomena such as polarization, and disagreement have been quantified and evaluated for the unsigned FJ model~\cite{ChLiDe18,MuMuTs18,DaGoLe13,XuBaZh21}, the definition, expressions and algorithms for these quantities on signed graphs are still lacking, since the antagonistic relationships between individuals make it more complicated to measure these social phenomena.

In this section, we first introduce the measures for some relevant social phenomena based on the signed FJ model and express them in terms of quadratic forms and the $\ell_2$ norms of vectors. Then, we design a signed Laplacian solver and use it to propose an efficient algorithm to approximate these social phenomena on signed graphs.


\subsection{Quantitative Measures for Social Phenomena}\label{Subsec:Meas}

We focus on some common phenomena in real social systems, such as conflict, polarization, and disagreement. Although quantitative indices and computation for these phenomena  have been extensively studied for unsigned graphs, their counterparts for signed graphs require careful reconsideration due to the existences of negative social links. Here we introduce the measures to quantify these phenomena based on the FJ model on signed graphs.


As in unsigned graphs~\cite{ChLiDe18}, we  define the internal conflict of the signed FJ model as follows.
\begin{definition}
For the FJ model on a signed graph $\calG = (V,E,w)$ and the initial opinion vector $\sss$, the internal conflict $I(\calG,\sss)$ is the sum of squares of the differences between the internal and expressed opinions over all nodes:
\begin{equation*}
I(\calG,\sss) = \sum_{i \in V} (\zz_i-\sss_i)^2.
\end{equation*}
\end{definition}

Besides internal conflict, individuals also suffer from the social cost incurred by their neighbours, as they prefer to express  similar opinions as their friends or opposite opinions to their opponents. This social cost is called external conflict or disagreement on unsigned graphs~\cite{ChLiDe18,MuMuTs18,DaGoLe13}. 
Below we extend the measure of the disagreement to the signed FJ model.

\begin{definition}
For the FJ model on a signed graph $\calG = (V,E,w)$ with an initial opinion vector $\sss$, its disagreement (or external conflict) $D(\calG,\sss)$ is defined as
\begin{equation*}
D(\calG,\sss) = \sum_{(i,j) \in E} (\zz_i-\AA_{i,j}\zz_j)^2.
\end{equation*}
\end{definition}

Note that at time $t$ the stress or psycho-social cost function of node $i$ is $(\zz_i(t)-\sss_i)^2+\sum_{j\in N_i}(\zz_i(t)-\AA_{i,j}\zz_j(t))^2$. At every time step, each node updates its expressed opinion with an aim to minimize its stress~\cite{BiKlOr15,HeZhLiRu20,RaHo21}. At the equilibrium, the total social cost, defined as the sum of the cost function over all nodes is $I(\calG,\sss) +D(\calG,\sss)$.

Since for a node $i$, each neighbour is either a friend or an enemy, its external social cost can be decomposed into two components: the cost $\sum_{j\in N^F_i} (\zz_i(t)-\zz_j(t))^2$ incurred from disagreeing with friends and the cost $\sum_{j\in N^E_i} (\zz_i(t)+\zz_j(t))^2$ incurred from agreeing with opponents. Hence, we give two variants of the disagreement for the FJ model on a signed graph $\calG$: disagreement with friends denoted by $F(\calG,\sss)$, and agreement with opponents denoted by $E(\calG,\sss)$.
\begin{definition}
For the FJ model on a signed graph $\calG = (V,E,w)$ with an initial opinion vector $\sss$, the disagreement with friends $F(\calG,\sss)$ is the sum of squares of the differences between
expressed opinions over all pairs of friends:
\begin{equation*}
F(\calG,\sss) = \sum_{(i,j) \in E^+} (\zz_i-\zz_j)^2.
\end{equation*}
\end{definition}

\begin{definition}
For the FJ model on a signed graph $\calG = (V,E,w)$ with an initial opinion vector $\sss$, the agreement with opponents $E(\calG,\sss)$ is the sum of squares of the sums of expressed opinions over all pairs of  opponent nodes:
\begin{equation*}
E(\calG,\sss) = \sum_{(i,j) \in E^-} (\zz_i+\zz_j)^2.
\end{equation*}
\end{definition}


Polarization can be thought of as the degree to which expressed opinions deviate from neutral  opinions, i.e., opinion value 0. In the following, we introduce the metric  of polarization~\cite{MaTeTs17,RaHo21}.

\begin{definition}
For the FJ model on a signed graph  $\calG = (V,E)$ with an initial opinion vector $\sss$, the polarization is defined as:
\begin{equation*}
P(\calG,\sss) = \frac{1}{n}\sum_{i \in V} \zz_i^2.
\end{equation*}
\end{definition}

According to their definitions, we can explicitly represent the aforementioned social phenomena  in terms of quadratic forms and the $\ell_2$ norms of vectors, as summarized in Lemma~\ref{lem:l2}.

\begin{lemma}\label{lem:l2}
For the FJ model on a signed graph $\calG=(V,E,w)$ with an initial opinion vector $\sss$, the internal conflict $I(\calG,\sss)$, disagreement $D(\calG,\sss)$, disagreement with friends $F(\calG,\sss)$, agreement with opponents $E(\calG,\sss)$, and polarization $P(\calG,\sss)$ can be conveniently expressed in terms of quadratic forms and the $\ell_2$ norms  as:
\begin{align*}
   I(\calG,\sss)& =  \zz^\top \LL^2 \zz = \sss^\top (\II+\LL)^{-1} \LL^2 (\II+\LL)^{-1} \sss \\
&=\norm{\LL(\II+\LL)^{-1}\sss}^2_2, \\
 F(\calG,\sss)&= \zz^\top \LL^+ \zz = \sss^\top (\II+\LL)^{-1} \LL^+ (\II+\LL)^{-1} \sss \\
&= \norm{\BB^+ (\II+\LL)^{-1}\sss}_2^2,\\
 E(\calG,\sss)&= \zz^\top \LL^- \zz = \sss^\top (\II+\LL)^{-1} \LL^- (\II+\LL)^{-1} \sss \\
&= \norm{\BB^- (\II+\LL)^{-1}\sss}_2^2,\\
  D(\calG,\sss)&= \zz^\top \LL \zz = \sss^\top (\II+\LL)^{-1} \LL (\II+\LL)^{-1} \sss \\
&= \norm{\BB (\II+\LL)^{-1}\sss}_2^2,\\
  P(\calG,\sss) &= \frac{1}{n}\zz^\top \zz=  \frac{ 1}{n}\sss (\II+\LL)^{-2} \sss = \frac{1}{n} \norm{ (\II+\LL)^{-1}\sss}_2^2.
\end{align*}
\end{lemma}

It is easy to verify that these quantities satisfy the following conservation law:
\begin{align*}
    I(\calG,\sss)+2D(\calG,\sss)+n P(\calG,\sss) &=\sum_{i=1}^n \sss_i^2, \\
     F(\calG,\sss)+ E(\calG,\sss) &= D(\calG,\sss),
\end{align*}
which extends the result in~\cite{MuMuTs18} for  unsigned graphs to  signed graphs.


\subsection{Signed Laplacian Solver}

As shown in Lemma~\ref{lem:l2}, directly calculating the social phenomena concerned involves inverting matrix $\II+\LL$, which takes $O(n^3)$ time and is impractical for large graphs. Although some efficient techniques have been proposed for unsigned graphs~\cite{XuBaZh21,XuZhGuZhZh22}, they cannot apply to signed graphs in a straightforward way. 

In this subsection, we first establish a connection between the absorbing probabilities for random walks on a signed graph and the absorbing probabilities for random walks on an unsigned graph, which is associated with the signed graph. Then, based on this connection we present a signed Laplacian solver, which allows for a fast approximation of relevant social phenomena with proven error guarantees.

For a signed graph $\calG=(V,E,w)$ with $n$ nodes and $m$ edges, we can define an unsigned graph $\hat{\calG}=(\hat{V},\hat{E})$ with $2n$ nodes and $2m$ edges. For the associated unsigned graph $\hat{\calG}$, we label its $2n$ nodes as $1,2,\ldots,2n$. The edge set $\hat{E}$ of $\hat{\calG}$ is constructed as follows. If there is a positive edge $(i,j)\in E$, then there are two edges $(i,j)$ and $(i+n,j+n)$ in $\hat{E}$. If there is a negative edge $(i,j)\in E$, then there are two edges $(i,j+n)$ and $ (i+n,j)$ in $\hat{E}$.

For the unsigned graph $\hat{\calG}=(\hat{V},\hat{E})$, we can define its augmented unsigned graph $\hat{\calH}=(\hat{X},\hat{R})$ with absorbing states. The construction details of graph $\hat{\calH}=(\hat{X},\hat{R})$ are as follows.
\begin{enumerate}[leftmargin=12pt,topsep=6pt]
    \item The node set $\hat{X}$ is defined as $\hat{X}=\hat{V} \cup \tilde{V}$, where $\tilde{V}$ is a set of $2n$ nodes such that for each node $i\in \hat{V}$, there is a copy node $\eta(i) \in \tilde{V}$ of node $i$;
    \item The edge set $\hat{R}$ includes all the edges $\hat{E}$ of $\hat{\calG}$, plus a new set of edges between each node $i\in \hat{V}$ and the copy node $\eta(i) \in \tilde{V}$ of node $i$. That is, $R=\hat{E} \cup \tilde{E}$, where $\tilde{E}=\{(i,\eta(i))| i\in \hat{V} \}$.
\end{enumerate}

We now define an absorbing random walk on the augmented graph $\hat{\calH}$ with $4n$ nodes, where the $2n$ nodes in $\hat{V}$ are transient nodes and the $2n$ nodes in $\tilde{V}$ are absorbing nodes. Then, the transition matrix $\PP'$ of this absorbing random walk is written as $\PP'=$
\begin{small}
\begin{equation*}
    \begin{bmatrix}
        (\II+\DD)^{-1} \AA^+ &  -(\II+\DD)^{-1} \AA^-& (\II+\DD)^{-1} & \OO\\
        -(\II+\DD)^{-1} \AA^- &  (\II+\DD)^{-1} \AA^+& \OO&(\II+\DD)^{-1} \\
         \OO &\OO & \II &\OO \\
         \OO&\OO & \OO& \II
    \end{bmatrix},
\end{equation*}
\end{small}

Define matrix $\QQ$ as $\QQ=\begin{bmatrix}
        (\II+\DD)^{-1} \AA^+ &  -(\II+\DD)^{-1} \AA^-\\
        -(\II+\DD)^{-1} \AA^- &  (\II+\DD)^{-1} \AA^+
    \end{bmatrix}$, and define matrix $\TT$ as $\TT=\begin{bmatrix}
         (\II+\DD)^{-1} & \OO\\
        \OO&(\II+\DD)^{-1}
    \end{bmatrix}$.
For a transient state $i\in \hat{V}$ and an absorbing state $\eta(j)\in \tilde{V}$, let $p'_{i,\eta(j)}$ be the absorbing probability of a random walk starting from $i$ being absorbed by $\eta(j)$. Define a $2n \times 2n$ matrix $\SS$ as
\begin{equation}\label{equ:2n2n}
    \SS=\begin{bmatrix}
  \II+ \DD- \AA^+ & \AA^- \\
    \AA^- & \II+\DD-\AA^+
\end{bmatrix}.
\end{equation}
Clearly, $\SS$ is symmetric, diagonally-dominant (SDD), since it can be expressed as $\SS=\II+\LL(\hat{\calG})$, where $\LL(\hat{\calG})$ is the Laplacian matrix of the unsigned graph $\hat{\calG}$. Then, $p'_{i,\eta(j)}$ can be represented as  $p'_{i,\eta(j)}=\ee_i^\top (\II-\QQ)^{-1}\TT\ee_j=\ee_i^\top \SS^{-1}\ee_j$.

Since $\SS$ is a block matrix, according to the matrix inversion in block form, we can further derive an expression of $p'_{i,\eta(j)}$ by distinguishing four cases: (i) $i \in \{1,2,\ldots,n\}$ and  $ j\in \{1,2,\ldots,n\}$; (ii) $i \in \{1,2,\ldots,n\}$ and  $j \in \{n+1,n+2,\ldots,2n\}$; (iii) $i \in \{n+1,n+2,\ldots,2n\}$ and $ j\in \{1,2,\ldots,n\}$; and (iv) $i \in \{n+1,n+2,\ldots,2n\}$ and $j \in \{n+1,n+2,\ldots,2n\}$. For the first case $i,j\in \{1,2,\ldots,n\}$, we have $p'_{i,\eta(j)}=\frac{1}{2}\ee_i^\top \left((\II+\DD-\AA^++\AA^-)^{-1}+(\II+\DD-\AA^+-\AA^-)^{-1}\right)\ee_j$. For the second case $i\in \{1,2,\ldots,n\}$ and $j \in \{n+1,n+2,\ldots,2n\}$, we have $p'_{i,\eta(j)}=\frac{1}{2}\ee_i^\top \left((\II+\DD-\AA^++\AA^-)^{-1}-(\II+\DD-\AA^+-\AA^-)^{-1}\right)\ee_{j-n}$. For the remaining two cases, considering the symmetry of the matrix $\SS$, it is easy to verify that $p'_{i+n,\eta(j)}=p'_{i,\eta(j+n)}$ and $p'_{i+n,\eta(j+n)}=p'_{i,\eta(j)}$ hold for $i,j\in \{1,2,\ldots,n\}$.

Theorem~\ref{thm:randomwalk} shows that in order to determine the expressed opinion $\zz_i$ of node $i$, we can alternatively compute the absorbing probabilities $p_{i,\sigma(j)}$ and $q_{i,\sigma(j)}$ of positive and negative random walk on the augmented signed graph $\calH$ of the signed graph $\calG$. According to the above arguments, we establish a direct relationship between the absorbing probabilities for absorbing random walks on signed graph $\calH$ and unsigned  graph $\hat{\calH}$. Specifically, for any $i,j\in \{1,2,\ldots,n\}$, $p'_{i,\eta(j)}=p'_{i+n,\eta(j+n)}=p_{i,\sigma(j)}$ and $p'_{i,\eta(j+n)}=p'_{i+n,\eta(j)}=q_{i,\sigma(j)}$. These relations allow us to   provide an alternative expression for the equilibrium expressed opinion for the FJ model on a signed graph, by using a matrix associated with an unsigned graph.

\begin{remark}
It should be mentioned that Hendrickx~\cite{He14} introduced a lifting approach, which implicitly establishes a relationship for opinions  between the continuous-time DeGroot model on a signed graph and the DeGroot model on a corresponding unsigned graph with  twice the number of nodes in the signed graph. We consider the discrete-time  FJ model on signed graphs. Although  our work looks somewhat similar to 
the previous work~\cite{He14}, they differ in at least several crucial ways. Firstly, discrete-time algorithms are more suitable for practical implementation because continuous-time algorithms need infinite information transmission rate~\cite{VaTrAn21}. Secondly, since the  FJ model is an extension of  the DeGroot model~\cite{MaAb19}, our result is more general than existing work~\cite{He14}. Finally, our
technique and algorithm presented below are also applicable to the discrete-time DeGroot model on signed graphs. Thus, our work offers a slightly more general framework for studying discrete-time opinion dynamics on signed graphs.
\end{remark}

\begin{theorem}\label{thm:lft}
For the FJ model of opinion dynamics on a signed graph $\calG$ with an initial opinion vector $\sss$, the equilibrium opinion vector can be expressed as
\begin{small}
\begin{equation}\label{equ:new_z}
    \zz = \frac{1}{2}
\begin{bmatrix}
\II &-\II
\end{bmatrix}
\SS^{-1}
\begin{bmatrix}
\II \\ -\II
\end{bmatrix} \sss,
\end{equation}
  \end{small}

where matrix $\SS$ is defined by~\eqref{equ:2n2n}, which is the fundamental matrix for an unsigned graph $\hat{\calG}$ associated with the signed graph $\calG$.
\end{theorem}

Theorem~\ref{thm:lft} shows that to determine the equilibrium opinion $\zz $ for the FJ model on a signed graph $\calG$, we can alternatively evaluate $\SS^{-1} \begin{bmatrix} \sss \\ -\sss\end{bmatrix}$ denoted by $\sss'$, where $\SS$ is an SDD matrix, which in fact corresponds to a fundamental matrix for the FJ model on an unsigned graph $\hat{\calG}$, expanded from $\calG$. In order to compute $\sss'$, we resort  to the fast SDD linear system solver~\cite{SpTe14,CoKyMiPaJaPeRaXu14}. Specifically, we propose a signed solver by extending the SDD solver to the FJ model on signed graphs, which avoids computing the inverse of $\SS$ or $\II+\LL$, but has proven error guarantees for various problems defined on signed graphs. Before doing so, we first introduce the SDD solver.


\begin{lemma}\label{lem:solver}~\cite{SpTe14,CoKyMiPaJaPeRaXu14}
Given a symmetric positive semi-definite matrix $\KK \in \mathbb{R}^{n \times n}$ with $m$ nonzero entries, a vector $\bb \in \mathbb{R}^n$, and an error parameter $\delta > 0$, there exists a solver, denoted as $\aaa = \SDDMSolver(\KK, \bb, \delta)$, which returns a vector $\aaa \in \mathbb{R}^n$ such that $\norm{\aaa - \KK^{-1} \bb}_{\KK} \leq \delta \norm{\KK^{-1} \bb}_{\KK}$. The runtime of this solver is expected to be $\tilde{O}(m)$, where $\tilde{O}(\cdot)$ notation suppresses ${\rm poly}(\log n)$ factors.
\end{lemma}

Note that the SDD solver in Lemma~\ref{lem:solver}  has been previously applied to solve many problems for opinion dynamics~\cite{ZhZh21,XuBaZh21,ZhBaZh21}. This solver allows us to provide a detailed analysis for the signed Laplacian solver  described below, and has an error guarantee. In addition, as deep learning evolves, some other solvers based on gradient propagation have been also utilized in the study on opinion dynamics~\cite{CiGiBo23,MaMiTaTa21}.


\begin{lemma}\label{lem:signsolver}
Given a signed graph $\calG=(V,E,w)$ with Laplacian matrix $\LL$, a matrix $\SS$ defined in~\eqref{equ:2n2n}, a vector $\yy\in \mathbb{R}^n$, and an error parameter $\delta>0$, there is a signed Laplacian solver $\ff=\textsc{SignedSolver}(\II+\LL,\yy,\delta)=\frac{1}{2}\begin{bmatrix} \II &-\II\end{bmatrix}\cdot \SDDMSolver\left(\SS, \begin{bmatrix}\II \\ -\II\end{bmatrix}\yy, \delta\right)$, which returns a vector $\ff$ satisfying
\begin{equation}\label{equ:appIL}
    \norm{\ff - (\II+\LL)^{-1} \yy}_{\II+\LL} \leq \delta \norm{(\II+\LL)^{-1} \yy}_{\II+\LL}.
\end{equation}
This signed Laplacian solver runs in expected time $\tilde{O}(m)$, where $\tilde{O}(\cdot)$ notation suppresses the ${\rm poly} (\log n)$ factors.
\end{lemma}

\begin{proof}
In order to prove~\eqref{equ:appIL}, it is equivalent to prove
\begin{equation*}
\ff^\top (\II+\LL) \ff + \yy^\top (\II+\LL)^{-1} \yy - 2\ff^\top\yy \leq \delta^2 \yy^\top (\II+\LL)^{-1}\yy.
\end{equation*}
Let $\cc=\SDDMSolver\left(\SS, \begin{bmatrix}\II \\ -\II\end{bmatrix}\yy, \delta\right)$. Eliminating $\ff$ and $\II+\LL$ by $\cc$ and $\SS$ yields
\begin{align}\label{equ:toprove}
&\frac{1}{8}\cc^\top \begin{bmatrix}\II \\ -\II\end{bmatrix}\begin{bmatrix}\II &-\II\end{bmatrix}\SS\begin{bmatrix}\II \\ -\II\end{bmatrix}\begin{bmatrix}\II & -\II\end{bmatrix}\cc \notag\\
&+\frac{1-\delta^2}{2}\yy^\top \begin{bmatrix}\II &-\II\end{bmatrix} \SS^{-1}\begin{bmatrix}\II \\ -\II\end{bmatrix} \yy - \cc^\top \begin{bmatrix}\II \\ -\II\end{bmatrix} \yy \leq 0.
\end{align}
By definition of $\cc$, to prove~\eqref{equ:toprove}, we only need to prove
\begin{equation*}
\frac{1}{8} \cc^\top \begin{bmatrix}\II \\ -\II\end{bmatrix}\begin{bmatrix}\II &-\II\end{bmatrix}\SS\begin{bmatrix}\II \\ -\II\end{bmatrix}\begin{bmatrix}\II  &-\II\end{bmatrix}\cc - \frac{1}{2}\cc^\top \SS \cc \leq 0,
\end{equation*}
which can be recast as
\begin{equation}\label{equ:semiposi}
\cc^\top \left(4\SS -  \begin{bmatrix} \II & -\II\\ -\II & \II \end{bmatrix} \SS \begin{bmatrix} \II & -\II\\ -\II & \II \end{bmatrix}\right) \cc \geq 0.
\end{equation}
Plugging the expression for matrix $\SS$ in~\eqref{equ:2n2n} into~\eqref{equ:semiposi} gives the following equation:
\begin{equation*}
4\SS - \begin{bmatrix} \II & -\II\\ -\II & \II \end{bmatrix} \SS \begin{bmatrix} \II & -\II\\ -\II & \II \end{bmatrix} = \begin{bmatrix} 2\MM & 2\MM \\ 2\MM & 2\MM \end{bmatrix},
\end{equation*}
where $\MM = \II+\DD-\AA^++\AA^-$ is an SDDM matrix.
Rewrite  $\cc$ as  $\cc = \begin{bmatrix} \uu \\ \vvv\end{bmatrix}$. Then~\eqref{equ:semiposi} is
expressed as
\begin{equation*}
\uu^\top \MM \uu +\vvv^\top \MM \vvv +2 \uu^\top \MM\vvv= (\uu+\vvv)^\top \MM (\uu+\vvv)\geq 0,
\end{equation*}
which is true since matrix $\MM$ is a positive definite matrix.
Combining the above analyses completes the proof.
\end{proof}

\subsection{Fast Evaluation Algorithm}


Lemma~\ref{lem:signsolver} indicates that for those quantities on signed graphs having form $(\II+\LL)^{-1}\sss$, we can exploit the signed solver to significantly reduce the computational time. We next apply Lemma~\ref{lem:signsolver} to obtain approximations for the quantities defined in Section~\ref{Subsec:Meas}.
\begin{lemma}\label{lm1}
Given a signed graph $\calG=(V,E,w)$ with Laplacian matrix $\LL$, incident matrix $\BB$,  positive incident matrix $\BB^+$,  negative incident matrix $\BB^-$, and a  parameter $\epsilon \in (0, \frac{1}{2})$,  consider  the FJ model of opinion dynamics on $\calG=(V,E,w)$ with  the internal opinion vector $\sss$ and let $\qq=\textsc{SignedSolver}(\II+\LL,\sss,\delta)$, where
	\begin{align*}
	\delta \leq \min\left\{\frac{\epsilon}{3\sqrt{2n}}, \frac{\norm{\sss}_{\LL}}{6\sqrt{2}n^2}\epsilon,\frac{\norm{\LL\sss}_2}{12\sqrt{2}n^3}\epsilon,\frac{\norm{\sss}_2}{2\sqrt{2}n}\epsilon\right\}.
	\end{align*}
Then, the following relations hold:
	\begin{align}
\label{equ:1}	& \norm{(\II+\LL)^{-1}\sss}^2_2 \approx_{\epsilon} \norm{\qq}^2_2,\\
\label{equ:2}	& \norm{\BB(\II+\LL)^{-1}\sss}^2_2 \approx_{\epsilon} \norm{\BB\qq}^2_2,\\
  \label{equ:3}      &\norm{\LL(\II+\LL)^{-1}\sss}^2_2 \approx_{\epsilon} \norm{\LL\qq}^2_2,\\
 \label{equ:4}       &\len{\norm{ \BB^+\qq}_2^2 - \norm{ \BB^+(\II+\LL)^{-1}\sss}_2^2}\leq \epsilon,\\
 \label{equ:5}       &\len{\norm{ \BB^-\qq}_2^2 - \norm{ \BB^-(\II+\LL)^{-1}\sss}_2^2}\leq \epsilon.
	\end{align}
\end{lemma}

\begin{proof}
We prove this lemma in turn. We first prove~\eqref{equ:1}. By Lemma~\ref{lem:signsolver}, we obtain
	\begin{align*}
	\norm{\qq - (\II+\LL)^{-1}\sss}^2_{\II+\LL} \leq \delta^2 \norm{(\II+\LL)^{-1}\sss}_{\II+\LL}^2.
	\end{align*}	
The term on the left-hand side (lhs) is bounded by
	\begin{align*}
	\norm{\qq - (\II+\LL)^{-1}\sss}_{\II+\LL}^2 	\geq & \norm{\qq - (\II+\LL)^{-1}\sss}^2_2 \\
	\geq & \left|\norm{\qq}_2 - \norm{(\II+\LL)^{-1}\sss}_2\right|^2,
	\end{align*}
while the term on the right-hand side (rhs) is bounded by
	\begin{align*}
	\norm{(\II+\LL)^{-1}\sss}^2_{\II+\LL} \leq 2n\norm{(\II+\LL)^{-1}\sss}^2_2.
	\end{align*}
These two bounds together lead to
	\begin{align*}
	 \len{\norm{\qq}_2 - \norm{(\II+\LL)^{-1}\sss}_2}^2 \leq 2n\delta^2 \norm{(\II+\LL)^{-1}\sss}^2_2.
	\end{align*}
Considering $\delta \leq \frac{\epsilon}{3\sqrt{2n}}$, we get
	\begin{align*}
	\frac{\len{\norm{\qq}_2 - \norm{(\II+\LL)^{-1}\sss}_2}}{\norm{(\II+\LL)^{-1}\sss}_2} \leq \sqrt{2n\delta^2 } \leq \frac{\epsilon}{3}.
	\end{align*}
Using the condition $0 < \epsilon < \frac{1}{2}$, we have
	\begin{align*}
	(1-\epsilon ) \norm{(\II+\LL)^{-1}\sss}^2_2 \leq \norm{\qq}^2_2 \leq (1+\epsilon ) \norm{(\II+\LL)^{-1}\sss}_2^2,
	\end{align*}
which completes the proof of~\eqref{equ:1}.

Then, we prove~\eqref{equ:2}. By Lemma~\ref{lem:signsolver}, we have
	\begin{align*}
	\norm{\qq - (\II+\LL)^{-1}\sss}^2_{\II+\LL} \leq \delta^2\norm{(\II+\LL)^{-1} \sss}^2_{\II+\LL}.
	\end{align*}	
The lhs is bounded as
	\begin{align*}
	& \norm{\qq - (\II+\LL)^{-1} {\sss}}^2_{\II+\LL}
	\geq  \norm{\qq - (\II+\LL)^{-1}\sss}^2_{\LL}\\ = & \norm{ \BB\qq -  \BB(\II+\LL)^{-1}\sss}^2_2
	\geq  \len{\norm{ \BB\qq}_2 - \norm{ \BB(\II+\LL)^{-1}\sss}_2}^2,
	\end{align*}
while the rhs is bounded as
	\begin{align*}
	\norm{(\II+\LL)^{-1}\sss}_{\II+\LL}^2 \leq 2n\norm{(\II+\LL)^{-1}\sss}^2_2
	\leq  2n^2,
	\end{align*}
where the relation $|\sss_i|\leq 1$ is used.
These two obtained bounds give
	\begin{align*}
	 \len{\norm{ \BB\qq}_2 - \norm{ \BB(\II+\LL)^{-1}\sss}_2}^2\leq 2\delta^2 n^2.
	\end{align*}
On the other hand,
	\begin{align*}
	\norm{ \BB(\II+\LL)^{-1}\sss}^2_2
	\geq  \frac{1}{4n^2}\norm{\sss}^2_{\LL}.
	\end{align*}
Considering $\delta \leq  \frac{\norm{\sss}_{\LL}}{6\sqrt{2}n^2}\epsilon$, one obtains
	\begin{align*}
	\frac{\len{\norm{ \BB\qq}_2 - \norm{ \BB(\II+\LL)^{-1}\sss}_2}}{\norm{ \BB(\II+\LL)^{-1}\sss}_2}
	\leq  \sqrt{\frac{8\delta^2 n^4}{\norm{\sss}^2_{\LL}} } \leq \frac{\epsilon}{3}.
	\end{align*}
Using $0 < \epsilon < \frac{1}{2}$, we obtain
	\begin{align*}
	(1-\epsilon ) \norm{ \BB(\II+\LL)^{-1}\sss}^2_2
	\leq \norm{ \BB\qq}^2_2 \leq (1+\epsilon ) \norm{ \BB(\II+\LL)^{-1}\sss}^2_2,
	\end{align*}
 completing the proof of~\eqref{equ:2}.

Next, we prove~\eqref{equ:3}. Applying  Lemma~\ref{lem:signsolver}, we obtain
\begin{small}
 \begin{align*}
	\norm{\qq - (\II+\LL)^{-1}\sss}^2_{\II+\LL} \leq \delta^2 \norm{(\II+\LL)^{-1}\sss}^2_{\II+\LL}.
	\end{align*}	
The term on the lhs is bounded by
	\begin{align*}
	& \norm{\qq - (\II+\LL)^{-1}\sss}^2_{\II+\LL}
	\geq  \norm{\qq - (\II+\LL)^{-1}\sss}_2^2 \\
 \geq& \frac{1}{4n^2}\norm{\LL\qq -\LL(\II+\LL)^{-1}\sss}_2^2 \\
	\geq &\frac{1}{4n^2} \len{\norm{\LL\qq}_2 - \norm{\LL(\II+\LL)^{-1}\sss}_2}^2.
	\end{align*}
\end{small}
Using the above-proved relation $\norm{(\II+\LL)^{-1}\sss}_{\II+\LL}^2 	\leq  2n^2$,  we have
	\begin{align*}
	\len{\norm{\LL\qq}_2 - \norm{\LL(\II+\LL)^{-1}\sss}_2}^2 \leq 8\delta^2 n^4.
\end{align*}
	On the other hand,
	\begin{align*}
	\norm{\LL(\II+\LL)^{-1}\sss}^2_2
	\geq  \frac{1}{4n^2}\norm{\LL\sss}_2^2.
	\end{align*}
Combining the above-obtained inequalities and $\delta \leq \frac{\norm{\LL\sss}_2}{12\sqrt{2}n^3}\epsilon$ yields
	\begin{align*}
	\frac{\len{\norm{\LL\qq}_2 - \norm{\LL(\II+\LL)^{-1}\sss}_2}}{\norm{\LL(\II+\LL)^{-1}\sss}_2}
	\leq  \sqrt{\frac{32\delta^2 n^6 }{\norm{\LL\sss}_2^2} } \leq \frac{\epsilon}{3}.
	\end{align*}
Considering $0 < \epsilon < \frac{1}{2}$, we derive
	\begin{align*}
(1-\epsilon ) \norm{\LL(\II+\LL)^{-1}\sss}_2^2 \leq \norm{\LL\qq}_2^2
	\leq(1+\epsilon ) \norm{\LL(\II+\LL)^{-1}\sss}_2^2,
	\end{align*}
which finishes the proof of~\eqref{equ:3}.

Finally, we prove the last two inequalities~\eqref{equ:4} and~\eqref{equ:5}. 	By Lemma~\ref{lem:signsolver}, we have
	\begin{align*}
	\norm{\qq - (\II+\LL)^{-1}\sss}^2_{\II+\LL} \leq \delta^2\norm{(\II+\LL)^{-1} \sss}^2_{\II+\LL}.
	\end{align*}	
The term on the lhs is bounded as
	\begin{align*}
	& \norm{\qq - (\II+\LL)^{-1} {\sss}}^2_{\II+\LL}
	\geq  \norm{\qq - (\II+\LL)^{-1}\sss}^2_{\LL^+}\\ = & \norm{ \BB^+\qq -  \BB^+(\II+\LL)^{-1}\sss}^2_2 \\
	\geq & \len{\norm{ \BB^+\qq}_2 - \norm{ \BB^+(\II+\LL)^{-1}\sss}_2}^2,
	\end{align*}
	while the term of the rhs is bounded as
	\begin{align*}
	\norm{(\II+\LL)^{-1}\sss}_{\II+\LL}^2 \leq 2n\norm{(\II+\LL)^{-1}\sss}^2_2
	\leq  2n^2\,.
	\end{align*}
 Combining these two bounds gives
	\begin{align*}
	 \len{\norm{ \BB^+\qq}_2 - \norm{ \BB^+(\II+\LL)^{-1}\sss}_2}^2\leq 2\delta^2 n^2.
	\end{align*}
Considering the fact that $\delta \leq \frac{\norm{\sss}_2}{2\sqrt{2}n}\epsilon$, we have
\begin{small}
 	\begin{align*}
	&   \len{\norm{ \BB^+\qq}_2^2 - \norm{ \BB^+(\II+\LL)^{-1}\sss}_2^2} \\
    = & |\left(\norm{ \BB^+\qq}_2 + \norm{ \BB^+(\II+\LL)^{-1}\sss}_2\right) \cdot\\
    &\left(\norm{ \BB^+\qq}_2 - \norm{ \BB^+(\II+\LL)^{-1}\sss}_2\right)| \\
     \leq & 2\sqrt{2}\delta n \norm{ \BB^+(\II+\LL)^{-1}\sss}_2
     \leq   2\sqrt{2}\delta n  \norm{\sss}_2 \leq \epsilon,
	\end{align*}
\end{small}
 which completes the proof of~\eqref{equ:4}.

Through replacing $\BB^+$ in the proof process of~\eqref{equ:4} by $\BB^-$, we can prove~\eqref{equ:5} in a similar way.
\end{proof}

Based on Lemmas~\ref{lem:signsolver} and~\ref{lm1}, we propose a nearly-linear time algorithm called \textsc{ApproxQuan}, which  approximates the internal conflict $I(\calG)$, disagreement $D(\calG)$, disagreement with friends $F(\calG)$, agreement with opponents $E(\calG)$, and polarization $P(\calG)$ for the FJ model on a signed graph. Algorithm~\ref{alg:VC} provides the details of algorithm \textsc{ApproxQuan}, the performance of which is summarized in Theorem~\ref{ThmV}.

\begin{algorithm}
\begin{small}
	\caption{$\textsc{ApproxQuan}\kh{\calG, \sss, \epsilon}$}
	\label{alg:VC}
	\Input{$\calG=(V,E,w)$: a signed graph;
 	$\sss$: initial opinion vector; $\epsilon$: the error parameter in $ (0, \frac{1}{2})$ \\
	}
	\Output{$\{ \tilde{I}(\calG), \tilde{D}(\calG),\tilde{F}(\calG),\tilde{E}(\calG),  \tilde{P}(\calG)\}$
	}
	$\delta=\min\left\{\frac{\epsilon}{3\sqrt{2n}}, \frac{\norm{\sss}_{\LL}}{6\sqrt{2}n^2}\epsilon,\frac{\norm{\LL\sss}_2}{12\sqrt{2}n^3}\epsilon,\frac{\norm{\sss}_2}{2\sqrt{2}n}\epsilon\right\}$ \;
	$\qq = \textsc{SignedSolver}(\II + \LL, \sss, \delta)$ \;
	$\tilde{I}(\calG,\sss) = \norm{\LL  \qq }^2_2$ \;
	$\tilde{D}(\calG,\sss) = \norm{ \BB \qq }^2_2$ \;
 	$\tilde{F}(\calG,\sss) = \norm{ \BB^+ \qq }^2_2$ \;
  	$\tilde{E}(\calG,\sss) = \norm{ \BB^- \qq }^2_2$ \;
	$\tilde{P}(\calG,\sss)= \norm{\qq }^2_2/n$ \;
	\textbf{return} $\{ \tilde{I}(\calG), \tilde{D}(\calG), \tilde{F}(\calG),\tilde{E}(\calG),  \tilde{P}(\calG) \}$ \;
 \end{small}
\end{algorithm}

\begin{theorem}\label{ThmV}
Given a signed undirected graph $\calG$, an error parameter $\epsilon \in (0, \frac{1}{2})$, and the internal opinion vector $\sss$,  the algorithm $\textsc{ApproxQuan}\kh{\calG, \sss, \epsilon}$ runs in expected time  $\tilde{O}(m)$, where $\tilde{O}(\cdot)$ notation suppresses the ${\rm poly} (\log n)$ factors, and returns approximations $\tilde{I}(\calG,\sss)$, $\tilde{D}(\calG,\sss)$, $\tilde{F}(\calG,\sss)$, $\tilde{E}(\calG,\sss)$, $\tilde{P}(\calG,\sss)$ for the  internal conflict $I(\calG,\sss)$,  disagreement  $D(\calG,\sss)$, disagreement with friends $F(\calG,\sss)$, agreement with opponents $E(\calG,\sss)$, and polarization $P(\calG,\sss)$, satisfying
	$\tilde{I}(\calG,\sss) \approx_\epsilon I(\calG,\sss)$,
	$\tilde{D}(\calG,\sss) \approx_\epsilon D (\calG,\sss)$,
        $|\tilde{F}(\calG,\sss)-F(\calG,\sss)| \leq \epsilon $,
        $|\tilde{E}(\calG,\sss)-E(\calG,\sss)| \leq \epsilon $, and
	$\tilde{P}(\calG,\sss) \approx_\epsilon P (\calG,\sss)$.
\end{theorem}

\section{Overall Opinion Optimization}


Except for relevant social phenomena, the overall opinion is another important quantity for opinion dynamics.
In this section, we first propose a problem of optimizing the overall expressed opinion for the signed FJ model by changing the initial opinions of a fixed number of nodes. We then provide an algorithm optimally solve the problem in $O(n^3)$ time. To reduce the running time, we also   design an efficient algorithm to approximately solve the problem in nearly-linear time.

\subsection{Problem Statement}

The overall expressed opinion is defined as the sum of expressed opinions $\zz_i$ of nodes $i\in V$ at equilibrium, which can be expressed as $\sum_{i=1}^n \zz_i=\boldsymbol{1}^\top (\II+\LL)^{-1}\sss$. This expression shows that the overall expressed opinion is influenced by both the internal opinion $\sss_i$ of each node and the network structure encoded in matrix $(\II+\LL)^{-1}$. These two factors determine together the opinion dynamics in the signed FJ model. Define vector $\hh=(\II+\LL)^{-1}\boldsymbol{1}$. Then the overall opinion is rewritten as $\boldsymbol{1}^\top (\II+\LL)^{-1}\sss=\hh^\top \sss = \sum_{i=1}^n \hh_i \sss_i$, where $\hh_i $ determines the extent to which the internal opinion $\sss_i$ of node $i$  contributes to the overall opinion. Note that $\hh_i $ is determined by the network structure, which is thus  called the structure centrality of node $i$ in the FJ model~\cite{Fr11}.

As shown above, the overall opinion is a function $g(\cdot)$ of the initial opinion $\sss$  and structure centrality $\hh$. Then, it can be expressed as $g(\sss)=\boldsymbol{1}^\top (\II+\LL)^{-1}\sss=\hh^\top \sss = \sum_{i=1}^n \hh_i\sss_i$, when the network structure is fixed. In this paper, we study the influence of initial opinions on the overall opinion, while keeping the network structure unchanged. Then, a natural problem arises, how to maximize the multi-variable objective function $g(\sss)$ by changing the initial opinions of a fixed number of nodes. Mathematically, the opinion maximization problem is formally stated as follows.

\begin{problem}[OpinionMax]\label{Pr-IOMi}
Given a signed graph $\calG = (V,E,w)$, an initial opinion vector $\sss$, and an integer $ k\ll n $, suppose that for each $i\in V$, its internal opinion $\sss_i$ is in the interval $[-1,1]$. The problem is how to optimally choose $k$ nodes and change their internal opinions, leading to a new initial opinion vector $\yy\in [-1,1]^n$, such that the overall opinion $g(\yy)$ is maximized under the constraint $\norm{\yy-\sss}_0 \leq k$.
\end{problem}

In a similar way, we can minimize the overall opinion by optimally changing the initial opinions of $k$ nodes, which is called problem \textsc{OpinionMin}. Note that both problem \textsc{OpinionMax} and problem \textsc{OpinionMin} are equivalent to each other. One can invert positive and negative signs of the  initial opinions to ascertain this equivalence. Thus, in what follows, we focus on problem \textsc{OpinionMax}.


It should be mentioned that a similar opinion maximization problem has been proposed in~\cite{XuHuWu20} by changing the initial opinions of an unfixed number of nodes. In the problem, the total amount of modification of the initial opinions has an upper bound~\cite{XuHuWu20}. In contrast, we focus on selecting a fixed number of nodes to change their initial opinions, with no constraints on the change of initial opinions, as long as they lie in the interval $[-1,1]$.

It is easy to show that for unsigned undirected graphs, increasing the internal opinion $\sss_i$ of any node $i$ leads to the increase of the overall equilibrium opinion. However, for signed graphs, increasing $\sss_i$ of node $i$ not necessarily results in an increase in the overall equilibrium opinion, as node $i$ may have  negative structure centrality $\hh_i$. According to the expression $g(\sss)= \sum_{i=1}^n \hh_i\sss_i$, we can draw the following conclusion for a node $i$. If $\hh_i >0$, increasing $\sss_i$ implies increasing the overall opinion; If $\hh_i <0$, decreasing $\sss_i$ leads to an increase of the overall opinion; If $\hh_i =0$, changing $\sss_i$ has no influence on the overall opinion. Moreover, it is not difficult to derive that for any node $i\in V$ with $\hh_i \neq 0$, changing $\sss_i$ in a proper way can result in an increase of the overall opinion, with the maximum increment being $c_i=|{\hh}_i|(1-\sss_i|{\hh}_i|/{\hh}_i )$ if $\hh_i \neq 0$ for any $i\in V$. Below we leverage this property to develop two algorithms solving the  problem \textsc{OpinionMax}.


\subsection{Optimal Solution}

Despite the combinatorial nature, the \textsc{OpinionMax} problem can be optimally solved as follows. We first compute $c_i=|{\hh}_i|(1- \sss_i|{\hh}_i|/{\hh}_i)$ for each node $i$ with nonzero $\hh_i$. Since $c_i$ is the largest marginal gain for node $i\in V$, we then select the $k$ nodes with the maximum value of $c_i$. Finally, we change the initial opinions of these $k$ selected nodes in the following way. If nodes have positive structure centrality, change their internal opinions to 1; otherwise change their internal opinions to -1.

Based on the above three operations, we design an algorithm to optimally solve the problem \textsc{OpinionMax}, which is outlined in Algorithm~\ref{al-optimal}.
This algorithm first computes the inverse of matrix $\II+\LL$ in $O(n^3)$ time. It then computes the vector $\hh$ in $O(n^2)$ time, and calculates $c_i$ for each $i \in V$ in $O(n)$ time. Finally, Algorithm~\ref{al-optimal} chooses $k$ nodes and modifies their internal opinions according to the structure centrality $\hh_i$ and the value $c_i$ for each candidate node $i$, which takes $O(n)$ time. Therefore, the overall time complexity of Algorithm~\ref{al-optimal} is $O(n^3)$. It should be mentioned that Algorithm~\ref{al-optimal}  returns the same results as the method in~\cite{XuHuWu20} when the number of selected  nodes is fixed.

\begin{algorithm}[tb]
\begin{small}
	\caption{\textsc{Optimal}$(\calG,\sss, k)$}
	\label{al-optimal}

	\Input{
		A graph $\calG=(V,E)$; an internal opinion vector $ \sss $; an integer $k$ obeying relation $1 \leq k \ll n$\\
	}
	\Output{
		$\yy$: A modified internal opinion vector with $\norm{\yy-\sss}_0 \leq k$
	}
	Initialize solution $\yy=\sss$ \;
	Compute $\hh = (\II+\LL)^{-1}\boldsymbol{1}$ \;
        \For{$ i\in V $}{
		\If{$\hh_i=0$}{$c_i \gets 0$}
            \Else{$c_i\gets |{\hh}_i|(1-\sss_i|{\hh}_i|/{\hh}_i )$}
            }
	\For{$ t=1 $ to $ k $}{
		Select $i$ s. t.  $i \gets \argmax_{i \in V} c_i$ \;
            \If{${\hh}_i=0$}{break\;}
            Update ${c}_i \gets 0$ \;
		Update solution $\yy_i \gets |{\hh}_i|/{\hh}_i$ \;	
	}	
	\Return $\yy$.
 \end{small}
\end{algorithm}

Since computing the vector $\hh$ for structure centrality  takes much time, Algorithm~\ref{al-optimal} is computationally unacceptable for large graphs. In the next subsection, we will design an efficient  algorithm based on the signed Laplacian solver \textsc{SignedSolver}.

\subsection{Fast Algorithm for Opinion Optimization}

To solve problem \textsc{OpinionMax} efficiently, using the signed Laplacian solver \textsc{SignedSolver} we propose a fast algorithm to approximate the structure centrality vector $\hh=(\II+\LL)^{-1}\boldsymbol{1}$ and solve the problem in nearly-linear time with respect to $m$, the number of edges. Let vector $\bar{\hh}$ be an approximation of $\hh$ returned by \textsc{SignedSolver}. The following lemma shows the relationship between elements in $\bar{\hh}$ and $\hh$.


\begin{lemma}\label{lem:approx_h}
Given a signed graph $\calG$ with Laplacian matrix $\LL$, a positive integer  $k$, and a parameter $\epsilon>0$, let $\bar{\hh}=\textsc{SignedSolver}( \II+\LL,\boldsymbol{1},\delta)$ be the approximation of $\hh=(\II+\LL)^{-1}\boldsymbol{1}$. Then the following inequality holds:
\begin{equation*}
    |\hh_i-\bar{\hh}_i| \leq \epsilon/(4k),
\end{equation*}
for any $\delta=\frac{\epsilon}{4k\sqrt{n+4m}}$.
\end{lemma}
\begin{proof}
    Let $\tilde{\hh}=\hh-\bar{\hh}$, then we have
$ \norm{\tilde{\hh}}_{\II+\LL}^2 \leq \delta^2 \norm{ \hh}_{\II+\LL}^2$.
Thus, we obtain that for any $i\in V$,
    \begin{equation*}
        \tilde{\hh}_i^2 \leq \tilde{\hh}_i^\top (\II+\LL)\tilde{\hh}_i \leq \delta^2 \hh^\top (\II+\LL) \hh \leq (n+4m)\delta^2 \leq \frac{\epsilon^2}{16k^2},
    \end{equation*}
 which completes the proof.
\end{proof}

Based on Lemma~\ref{lem:approx_h}, we can approximate each element of $\hh$ with an absolute error guarantee. Exploiting this lemma, we propose a fast algorithm to approximately solve  the problem \textsc{OpinionMax}, which is outlined in Algorithm~\ref{alg:fastopt}. The performance of this fast algorithm is stated in Theorem~\ref{thm:alg2}.

\begin{algorithm}[tb]
\begin{small}
	\caption{\textsc{ApproxOpin}$(\calG,\sss, k)$}
	\label{alg:fastopt}
	\Input{
		A signed $\calG=(V,E,w)$; an internal opinion vector $ \sss $; an integer $k$ obeying relation $1 \leq k \ll n$\\
	}
	\Output{
		$\yy$: A modified internal opinion vector with $\norm{\yy-\sss}_0 \leq k$
	}
        Set $\delta=\frac{\epsilon}{2k\sqrt{n+4m}}$ \;
	Initialize solution $\yy= \sss$ \;
	Compute $\bar{\hh} = \textsc{SignedSolver}( \II+\LL,\boldsymbol{1},\delta)$ \;
 \For{$ i\in V $}{
		\If{$\bar{\hh}_i=0$}{$\bar{c}_i \gets 0$}
            \Else{$\bar{c}_i \gets |\bar{\hh}_i|(1-\sss_i|\bar{\hh}_i|/\bar{\hh}_i )$}
            }
	\For{$ t=1 $ to $ k $}{
		Select $i$ s. t.  $i \gets \argmax_{i \in V} \bar{c}_i$ \;
            \If{$\bar{\hh}_i=0$}{break\;}
            Update $\bar{c}_i \gets 0$ \;
		Update solution $\yy_i \gets |\bar{\hh}_i|/\bar{\hh}_i$ \;	
	}	
	\Return $\yy$.
 \end{small}
\end{algorithm}

\begin{theorem}\label{thm:alg2}
 For  given parameters $k$ and $\epsilon$, algorithm \textsc{ApproxOpin} runs in time $\tilde{O}(m)$, and outputs a solution vector $\yy$ satisfying $|g(\yy^*)-g(\yy)| \leq \epsilon$, where $\yy^*$ is the optimal solution to problem \textsc{OpinionMax}.
\end{theorem}
\begin{proof}
Define $\bar{c}_i=|\bar{\hh}_i|(1-\sss_i|\bar{\hh}_i|/\bar{\hh}_i)$.
Using Lemma~\ref{lem:approx_h} and the relation  $c_i=|{\hh}_i|(1-|{\hh}_i|/{\hh}_i \sss_i)$, we suppose that inequality $|c_i-\bar{c}_i| <\epsilon/(2k)$ holds for any $i\in V$. Assume that sets $T_1$ and $T_2$ are returned by algorithms \textsc{Optimal} and \textsc{ApproxOpin}, respectively.
Then, we have
\begin{equation*}
    g(\yy^*)-g(\yy) =\sum_{i\in T_1} c_i -\sum_{j\in T_2} c_j \geq 0.
\end{equation*}
On the other hand, using $|c_i-\bar{c}_i| <\epsilon/(2k)$, we obtain
\begin{equation*}
    g(\yy^*)-g(\yy) = \sum_{i\in T_1} c_i -\sum_{j\in T_2} c_j \leq  \sum_{i\in T_1} \bar{c}_i -\sum_{j\in T_2} c_j + \epsilon/2 \leq \epsilon,
\end{equation*}
which completes the proof.
\end{proof}

\section{Experiments}
To evaluate the accuracy and efficiency of our algorithms \textsc{ApproxQuan} and \textsc{ApproxOpin} for two different tasks, we conduct extensive experiments on sixteen signed networks of different sizes. 

\subsection{Setup}

\textbf{Environment and repeatability.} We conduct all experiments using a single thread on a machine with a 2.4 GHz Intel i5-9300 CPU and 128GB of RAM. All algorithms  are realized  using the programming language \textit{Julia}. The parameter $\epsilon$ is set to be $10^{-5}$ for all experiments. Our code is publicly available at \url{https://github.com/signFJ/signFJ}.

\begin{table*}[htbp]
	\centering
		\caption{Statistics for networks and performance of algorithms \textsc{ApproxQuan} and \textsc{ApproxOpin}. The initial opinions obey a uniform distribution.}\label{tab}
\fontsize{7}{7}\selectfont			
\begin{tabular}{m{0.52cm}<{\centering}m{1.2cm}m{1.15cm}<{\raggedleft}m{1.17cm}<{\raggedleft}m{0.47cm}<{\centering}m{1.15cm}<{\centering}m{0.65cm}<{\centering}m{0.65cm}<{\centering}m{0.65cm}<{\centering}m{0.65cm}<{\centering}m{0.65cm}<{\centering}m{0.01cm}<{\centering}m{0.65cm}<{\centering}m{1.25cm}<{\centering}m{0.9cm}<{\centering}}
				\toprule
				\multirow{3}*{Type} &
				\multirow{3}*{Networks}  &\multirow{3}*{ Nodes} &\multirow{3}*{ Edges}  &   \multicolumn{7}{c}{Quantification of Social Phenomena} & & \multicolumn{3}{c}{Opinion Optimization}\\
                \cmidrule{5-11} \cmidrule{13-15}&&&&\multicolumn{2}{c}{ Time (seconds)}&\multicolumn{5}{c}{ Relative Error $(\times 10^{-8})$}   &&\multicolumn{2}{c}{ Time (seconds)} &\multirow{2}*{\parbox[t]{10mm}{\centering Relative \\Error\\$(\times 10^{-8})$}}\\
                \cmidrule{5-11}  \cmidrule{13-14} &&&&          	\textsc{Exact}& \textsc{ApproxQuan} & $I(\calG,\sss)$ &$D(\calG,\sss)$ & $F(\calG,\sss)$ & $E(\calG,\sss)$ &$P(\calG,\sss)$ &          &	\textsc{Optimal}& \textsc{ApproxOpin} \\
				\midrule
				\multirow{9}*{\parbox[t]{9mm}{\centering Original\\Signed\\Graphs}  }
               & Bitcoinalpha    &3,783& 24,186&1.92&0.38&0.26&0.01&0.02&0.37&0.02& &1.78 & 0.21&2.47\\
                &Bitcoinotc      &5,881& 35,592&4.40&0.40&1.70&0.08&0.14&0.95&0.15& &4.16 & 0.07&0.98\\
                &Wikielections    &7,118&103,675&9.66&0.18&0.42&0.03&0.03&0.04&0.04& &8.94 &0.02&0.37\\
                &WikiS      &9,211& 646,316&19.11&1.41&3.54&0.59&0.67&1.01&1.56& & 19.51&0.12&0.15\\  
                &WikiM      &34,404& 904,768&1276&2.42&7.15&0.18&0.02&1.36&0.23& &1091& 1.70&0.77\\  
                &SlashdotZoo&79,120& 515,397&--&1.32&--&--&--&--&--& &--& 1.60&--\\
                &WikiSigned &138,592& 740,397&--&1.70&--&--&--&--&--& &--& 1.66&--\\
                &Epinions   &131,828& 841,372&--&1.89&--&--&--&--&--& &--& 1.06&--\\
                &WikiL      &258,259& 3,187,096&--&7.03&--&--&--&--&--& &--& 6.59&--\\
                \midrule
                	\multirow{7}*{\parbox[t]{9mm}{\centering Modified\\Signed\\Graphs}  }
&PagesGovernment    &7,057&   89,455&9.31&0.72&1.88&6.94&8.41&0.04&0.91& &8.33& 0.40&1.57\\
&Anybeat            &12,645&   49,132&87.3&0.52&0.23&1.94&0.01&1.23&0.06& &85.9& 0.53&0.87\\
&Google            &875,713& 5,105,040&--&16.72&--&--&--&--&--& &--& 16.03&--\\
&YoutubeSnap      &1,134,890& 2,987,624&--&10.13&--&--&--&--&--& &--& 9.27&--\\
&Pokec          & 1,632,803& 30,622,564&--&108.03&--&--&--&--&--& &--& 97.57&--\\
&DBpediaLinks   &18,268,992&172,183,984&--&732.75&--&--&--&--&--& &--& 703.17&--\\
&FullUSA        &23,947,300& 57,708,600&--&186.72&--&--&--&--&--& &--& 170.69&--\\
				\bottomrule
			\end{tabular}
\end{table*}

\textbf{Datasets.}
We use sixteen datasets of two types of signed graphs: real-world original signed graphs and artificially modified real-world graphs. The real-world original signed graphs are from actual networks, for which the original signs are kept unchanged. The artificially modified graphs are generated from real unsigned graphs, by randomly assigning a negative sign to each edge in unsigned graphs with a probability of $0.3$. These network datasets are publicly available in KONECT~\cite{Ku13} and SNAP~\cite{LeSo16}, and their statistic is presented in Table~\ref{tab}, where networks are listed in increasing order of the number of nodes.

\textbf{Internal opinion distributions.} The initial opinion of each individual is important for computing social phenomenon measures  and optimizing the overall opinion. In practice, the initial opinion vector is generally not available, due to privacy concerns.  However, the initial opinions of nodes can be estimated~\cite{DaGoPaSa13}. Since obtaining initial opinions is outside the scope of this paper, in our experiments,
we assume that the initial opinion vector is known beforehand. We use three different distributions of initial opinions: uniform, exponential, and power-law, which are generated as follows. For the uniform distribution, the initial opinion of every node is generated uniformly in the range of $[-1, 1]$. For the exponential and power-law distributions, we first generate the initial opinions of all nodes in the range of $[0, 1]$ as in~\cite{XuBaZh21}. Then for every node, we change its initial opinion to its opposite number with probability of $0.5$.

\subsection{Performance of  Algorithm \textsc{ApproxQuan}}

We first evaluate the efficiency of our fast algorithm \textsc{ApproxQuan} for quantifying various social phenomena. For this purpose, we compare it with the exact algorithm, called \textsc{Exact}, which computes all relevant quantities by inverting the matrix $\II+\LL$. Table~\ref{tab} reports the running time of \textsc{ApproxQuan} and \textsc{Exact} on different networks. As shown in Table~\ref{tab}, the running time of \textsc{ApproxQuan} is always less than that of \textsc{Exact} for each of the considered networks with relatively small sizes. For networks with more than 40,000 nodes, \textsc{Exact} fails to run due to the high memory and time requirements. In contrast, \textsc{ApproxQuan} is able to approximate all the quantities in less than one thousand seconds. Moreover, \textsc{ApproxQuan} is scalable to massive networks with over 20 million nodes.

In addition to being highly efficient, algorithm \textsc{ApproxQuan} is also very accurate, compared with algorithm \textsc{Exact}. To show this, in Table~\ref{tab}, we compare the approximate results for \textsc{ApproxQuan} with the exact results for \textsc{Exact}. For each network, we compute the relative error $|\theta-\tilde{\theta}|/\theta$ for each quantity $\theta$ and its approximation $\tilde{\theta}$ returned by \textsc{ApproxQuan}. Table~\ref{tab} gives the relative errors for the five estimated quantities, including internal conflict $I(\calG)$, disagreement $D(\calG)$, disagreement with friends $F(\calG)$, agreement with opponents $E(\calG)$, and polarization $P(\calG)$. The results indicate that the actual relative errors for all quantities and networks are negligible, with all errors less than $10^{-7}$. Thus, \textsc{ApproxQuan} is not only highly fast but also highly effective in practice.

\begin{figure}[tbp]
	\centering
	\includegraphics[width=0.9\linewidth]{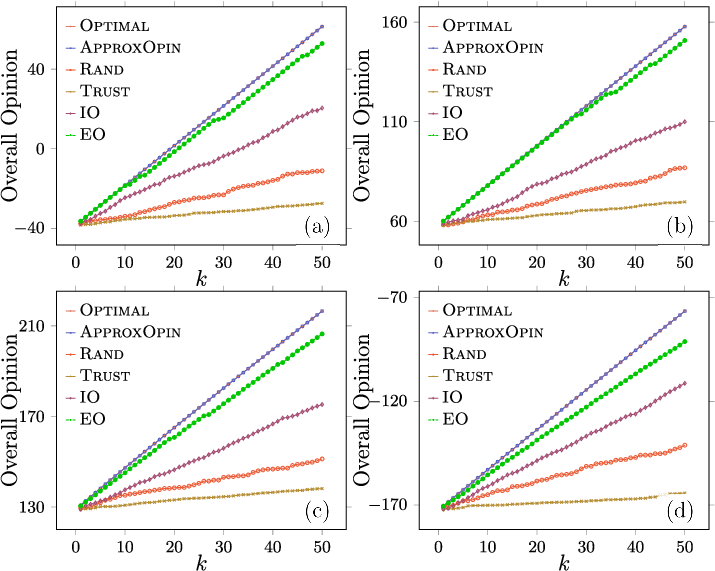}
	\caption{Overall opinions for  algorithms \textsc{ApproxOpin}, \textsc{Optimal}, and four baselines on four real networks: (a) Bitcoinalpha, (b) Wikielections, (c) WikiM, and  (d) Anybeat.}\label{fig:1}	
\end{figure}

\begin{figure}[tbp]
	\centering
	\includegraphics[width=0.9\linewidth]{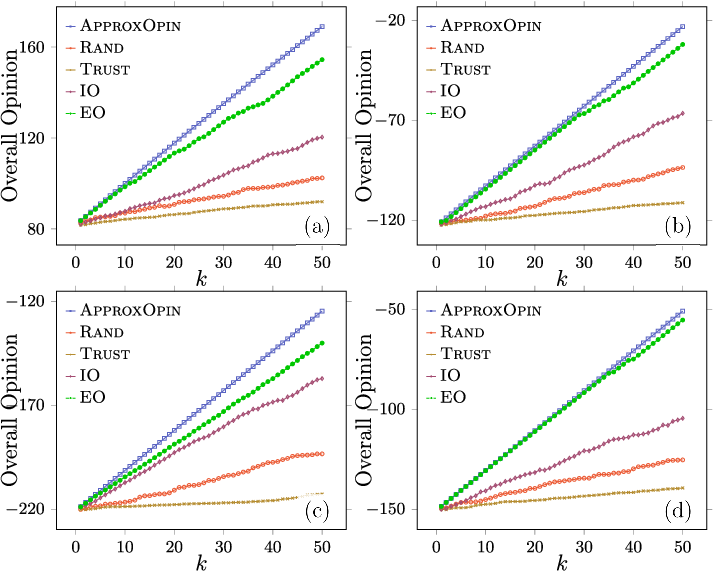}
	\caption{Overall opinions for  algorithms \textsc{ApproxOpin} and four baselines on four real networks: (a) Epinions, (b) WikiL, (c) Pokec, and  (d) FullUSA.}\label{fig:2}	
\end{figure}

\subsection{Performance of Algorithm \textsc{ApproxOpin}}

We continue to evaluate the performance of \textsc{ApproxOpin}. To achieve our goal, we compare algorithm \textsc{ApproxOpin} with the optimal algorithm \textsc{Optimal} and four baselines as in~\cite{XuHuWu20}, including \textsc{Rand}~\cite{LiChWaZh13}, \textsc{Trust}~\cite{ChFaLiFeTaTa15}, \textsc{IO}~\cite{MuMuTs18}, and \textsc{EO}~\cite{ChLiDe18}. \textsc{Rand} randomly selects $k$ nodes and changes their initial opinions to 1. \textsc{Trust} selects $k$ nodes with the largest differences between the numbers of friends and opponents, and changes their initial opinions to 1. \textsc{IO} selects $k$ nodes with the lowest internal opinions and changes their initial opinions to 1. \textsc{EO} selects $k$ nodes with the lowest expressed opinions,  and changes their initial opinions to 1.

We first assess the effectiveness of \textsc{ApproxOpin}. We change the initial opinions of $k=1,2,\ldots,50$ nodes by using \textsc{ApproxOpin}, \textsc{Optimal}, and the four baseline approaches. 
Figure~\ref{fig:1} illustrates the comparison of overall opinion for these methods on four small networks with less than 40,000 nodes, since for networks with over 40,000 nodes, \textsc{Optimal} fails to run. We observe that for these small networks, \textsc{ApproxOpin} consistently returns results that are close to the optimal solutions, both of which outperform the other four baselines. To further demonstrate the accuracy of \textsc{ApproxOpin}, in Table~\ref{tab}, we compare the relative error for the gain of the overall opinion for \textsc{ApproxOpin}  with respect to that for \textsc{Optimal} on seven small networks with $k= 50$. As shown in Table~\ref{tab}, the relative errors are all less than $10^{-7}$, indicating the high similarity of the results obtained by \textsc{ApproxOpin} and \textsc{Optimal}. We also compare \textsc{ApproxOpin} with the baseline strategies on four relatively large networks with over 40,000 nodes, and report the results in Figure~\ref{fig:2}, which again indicates that \textsc{ApproxOpin} is much better than the four baselines.

With regard to the efficiency, in Table~\ref{tab}, we compare the running time of \textsc{ApproxOpin} and \textsc{Optimal} on different networks for $k=50$. As shown in Table~\ref{tab}, \textsc{ApproxOpin} is significantly faster than \textsc{Optimal}, especially when networks become larger. Particularly, \textsc{Optimal} fails to run on networks with more than 40,000 nodes, while \textsc{ApproxOpin} can still run efficiently, which is even scalable to massive networks with more than twenty million nodes.

\section{Related Work}


\textbf{FJ model for opinion dynamics.} 
The FJ model~\cite{FrJo90} is a popular model for opinion dynamics, which has been extensively studied on unsigned graphs. For example, the sufficient condition for stability was studied in~\cite{RaFrTeIs15}, the formula for the equilibrium expressed opinion was derived in~\cite{DaGoPaSa13,BiKlOr15}, and the interpretations were provided from different aspects in~\cite{GiTeTs13,GhSr14,BiKlOr15,XuBaZh21}. Besides, by incorporating different aspects affecting opinion evolution and formulation, many variants of the FJ model have been proposed, including peer pressure~\cite{SeGrSqRa19}, stubbornness~\cite{XuZhGuZhZh22}, interactions among higher-order neighbours~\cite{ZhXuZhCh20}, and so on. Most prior works for the FJ model are based on unsigned graphs, which capture only the positive or cooperative relationships between individuals, ignoring the antagonistic or competitive relationships. Very recently, the FJ model was extended to the signed graphs, which incorporate both cooperative and competitive relationships~\cite{HeZhLiRu20,XuHuWu20,RaHo21,HeZeZhLi22}. For the FJ model on signed graphs, some relevant problems have been addressed, including the convergence criteria~\cite{HeZhLiRu20}, explanation of opinion update~\cite{RaHo21}, and opinion maximization by changing initial opinion~\cite{XuHuWu20}. However, the interpretation for expressed opinions is still lacking.

\textbf{Quantification and algorithms for social phenomena.}
The explosive growth of social media and online social networks produces diverse social phenomena, such as polarization~\cite{MaTeTs17,MuMuTs18,AmBoSi19}, disagreement~\cite{MuMuTs18}, filter bubbles~\cite{BaChLa23,La22}, conflict~\cite{ChLiDe18}, and controversy~\cite{ChLiDe18}, to name a few.
In addressing these challenges, research has evolved in different directions.
Fast algorithms were proposed to efficiently compute these quantities~\cite{XuBaZh21,XuZhGuZhZh22}. Some studies tried to find user groups open to ``counter-information''~\cite{GeMiYoZe18,FaBaGe20} and tried to connect users with opposing views, hoping to lessen filter bubble effects~\cite{ToRoGo21,ZhBaZh21,AmSi19,MuMuTs18}.  
 Moreover, the exploration of using influence models in social media to combat filter bubbles has also gained traction~\cite{TuAsGi20,PiCe19,MaAsGaGi20,GaGiPaTa17}.
These measures and algorithms for social phenomena tend not to apply to signed graphs. To make up for this deficiency, we extend these measures for the FJ model to signed graphs. Due to the incorporation of competitive relationships, previous approximation algorithms~\cite{XuBaZh21,XuZhGuZhZh22} are not suitable for signed graphs anymore. This motivates us to present a nearly linear time algorithm for estimating these quantities on signed graphs.

\textbf{Optimization of overall opinion.}
Various schemes have been proposed to maximize or minimize the overall opinion based on different models for opinion dynamics. On the basis of the DeGroot model, many groups have addressed the problem of maximizing the overall opinion by leader selection~\cite{YiCaPa21,HuZa22,ZhZh23} or link suggestion~\cite{ZhZh21,ZhZhLiZh23}. Based on independent cascade and linear threshold models, a similar problem, called the influence maximization problem, has also been studied~\cite{KeKlTa03,TaXiSh14,TaShXi15,TaTaXiYu18}. The opinion optimization problem has also attracted extensive attention for the FJ model.  In the past decade, different node-based strategies have been applied to optimize the overall opinion of the FJ model on unsigned graphs, including modifying initial opinions~\cite{AhDeHaMaYa15}, expressed opinions~\cite{GiTeTs13}, and susceptibility to persuasion~\cite{AbKlPaTs18,ChLiSo19,AbChKlLiPaSoTs21,MaMiTaTa21}. Most previous research focused on opinion optimization on unsigned graphs, with the exception of a few work~\cite{XuHuWu20}. In~\cite{XuHuWu20}, the problem of opinion optimization on signed graphs was studied by changing initial and external opinions, and two algorithms were developed  with complexity $O(n^{3})$, which are computationally infeasible for large graphs. Although we address a similar problem, our algorithm is efficient and effective, with nearly-linear time complexity and a proven error guarantee compared to the optimal solution.

\textbf{Research on signed graphs.}
A concerted effort has been devoted to delving into diverse aspects of signed graphs, including finding conflicting groups~\cite{TzOrGi20}, exploring polarization~\cite{XiOrGi20}, detecting communities~\cite{BoGaGiOrRu19, SuChWaZhWa20}, and so on. Moreover, clique computation and enumeration have also attracted much attention. Different algorithms have been proposed for computing maximum structural balanced cliques~\cite{YaChQi22}, enumerating maximal balanced bicliques~\cite{SuWuChWaZhLi22}, and searching signed cliques~\cite{LiDaQiWaXiYuQi19} on signed graphs. Finally, the influence diffusion process and influence maximization problem have also been studied on signed networks~\cite{LiLi19, YiHuChYuLi19, KaKhDaKuKh22}. However, the methods for studies on signed graphs are not applicable to the FJ model defined on signed networks.

\section{Conclusion}

In this paper, we studied the Friedkin-Johnsen (FJ) model for opinion dynamic on a signed graph. We first interpreted the equilibrium opinion of every node by expressing it in terms of the absorbing probabilities of a defined absorbing random walk on an augmented signed graph. We then quantified some relevant social phenomena and represented them as the $\ell_2$ norms of vectors. Moreover, we proposed a signed Laplacian solver, which approximately evaluates these quantities in nearly-linear time but has error guarantees. We also considered the problem of opinion optimization by modifying the initial opinions of a fixed number of nodes, and presented two algorithms to solve this problem. The first algorithm optimally solves the problem in cubic time, while the second algorithm provides an approximation solution with an error guarantee in nearly-linear time. Extensive experiments on real signed graphs demonstrate the effectiveness and efficiency of our approximation algorithms.

It deserves to mention that although we focus on unweighted signed graphs,  our analyses and algorithms for the  FJ opinion dynamics model can be  extended to  weighted signed graphs.
Future work includes the applications of our signed Laplacian solver to other problems for the FJ model on signed graph, such as optimizing disagreement, conflict, and polarization under different constraints, or maximizing (or minimizing) the overall opinion by using other strategies different from that used in this paper.


\bibliographystyle{IEEEtran}
\normalem
\bibliography{revision,signFJ,newref,expressedopinion,hyper,kedges}

\begin{IEEEbiography}[{\includegraphics[width=1in,height=1.25in,clip,keepaspectratio]{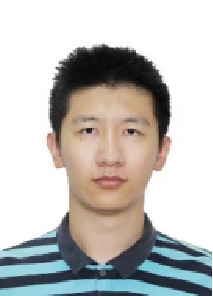}}]
{Xiaotian Zhou}
received the B.S. degree in mathematics science from Fudan University, Shanghai, China, in 2020. He is currently pursuing the PhD's degree in School of Computer Science, Fudan University, Shanghai, China. His research interests include network science, computational social science, graph data mining, and social network analysis.
\end{IEEEbiography}

\begin{IEEEbiography}[{\includegraphics[width=1in,height=1.25in,clip,keepaspectratio]{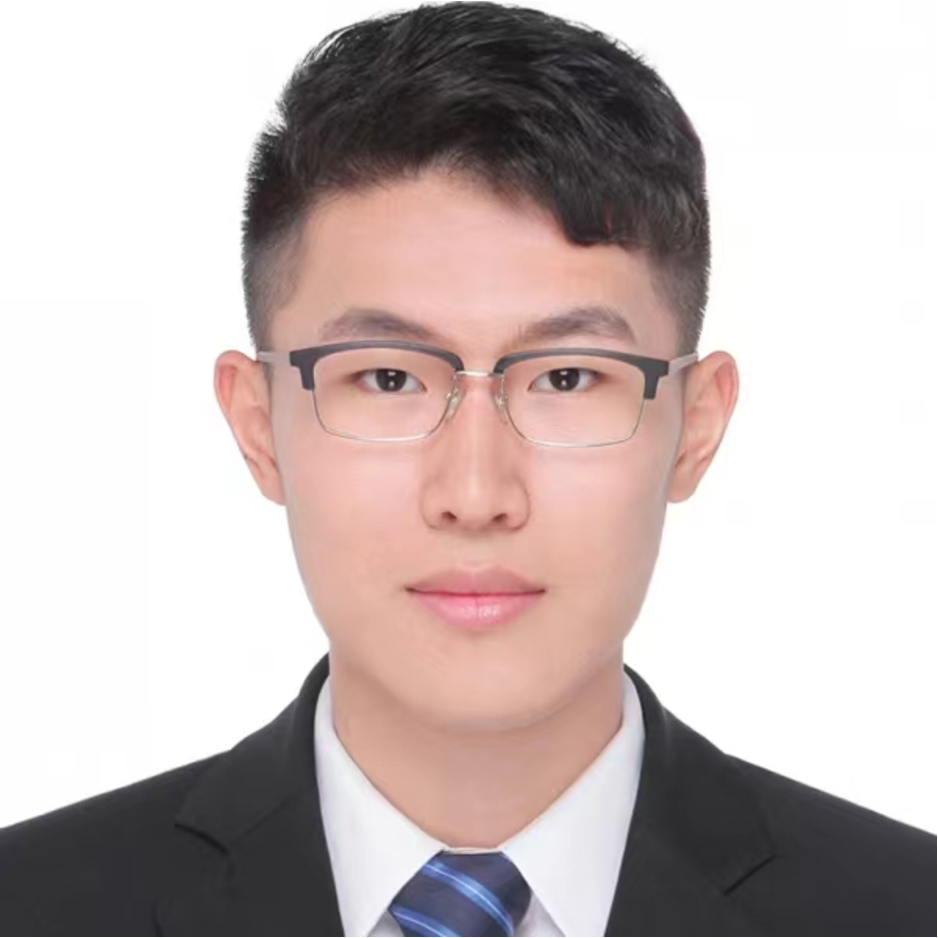}}]
{Haoxin Sun}
received the B.S. degree in mathematics science from Fudan University, Shanghai, China, in 2021. He is currently pursuing the PhD's degree in School of Computer Science, Fudan University, Shanghai, China. His research interests include network science, computational social science, graph data mining, and social network analysis.
\end{IEEEbiography}

\begin{IEEEbiography}[{\includegraphics[width=1in,height=1.25in,clip,keepaspectratio]{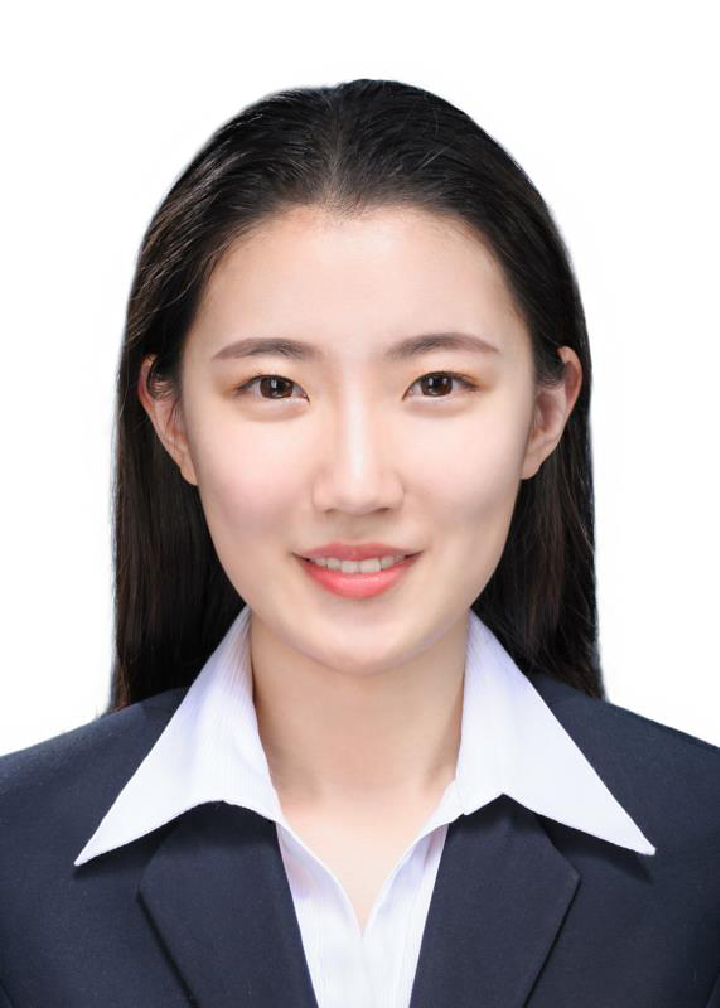}}]{Wanyue Xu}
	received the B.Eng. degree in School of Mechanical, Electrical, and Information Engineering, Shangdong University, Weihai, China, in 2019. She is currently pursuing  the Ph.D.	degree in School of Computer Science, Fudan University, Shanghai, China.  Her
	research interests include network science, graph
	data mining, social network analysis, and random
	walks. She is a student member of the IEEE.
\end{IEEEbiography}

\begin{IEEEbiography}
[{\includegraphics[width=1in,height=1.25in,clip,keepaspectratio]{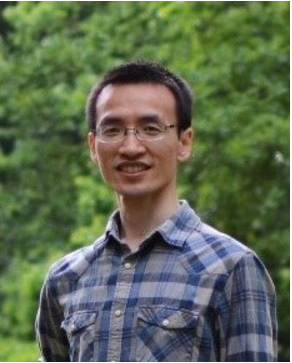}}]{Wei Li}
received the B.Eng. degree in automation and the M.Eng. degree in control science and engineering from the Harbin Institute of Technology, China, in 2009 and 2011, respectively, and the Ph.D. degree from the University of Sheffield, U.K., in 2016. After being a research associate at the University of York, UK, he is currently an associate professor with the Academy for Engineering and Technology, Fudan University. His research interests include robotics and computational intelligence, and especially self-organized/swarm systems, and evolutionary machine learning.
\end{IEEEbiography}

\begin{IEEEbiography}[{\includegraphics[width=1in,height=1.25in,clip,keepaspectratio]{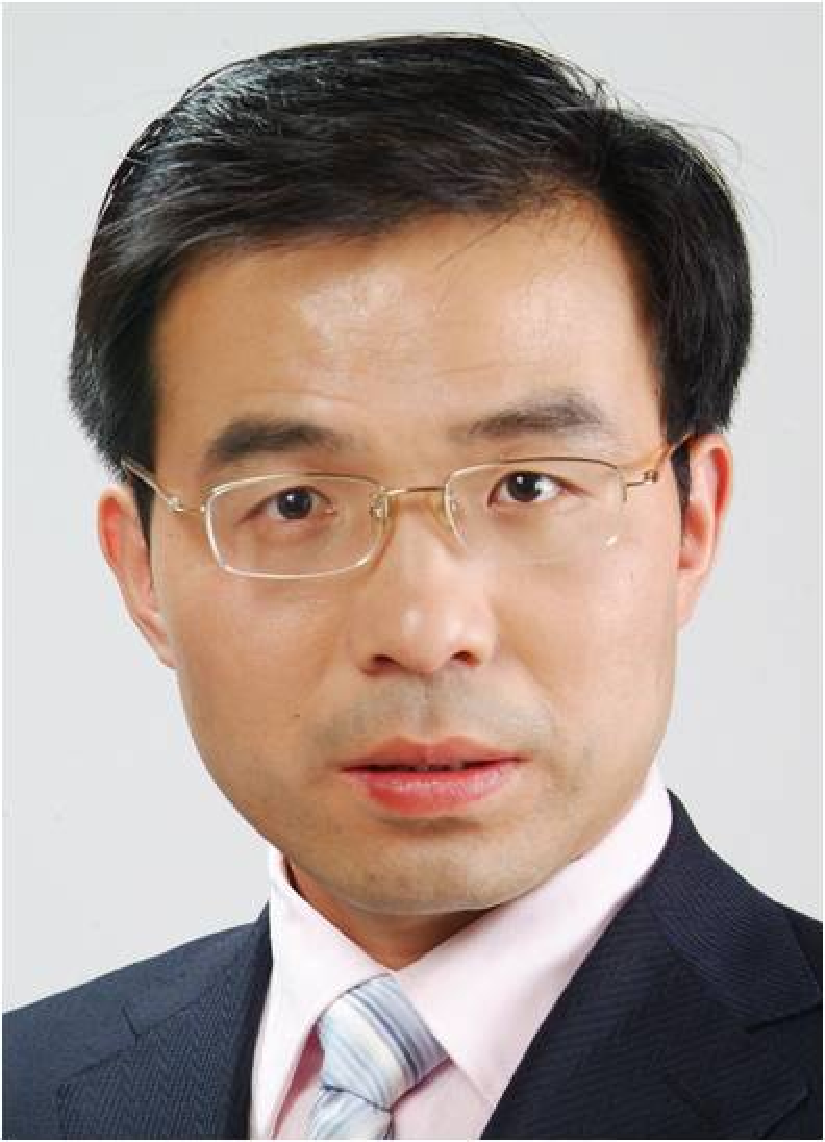}}]
{Zhongzhi Zhang}
	(M'19)	 received the B.Sc. degree in applied mathematics from Anhui University, Hefei, China, in 1997 and the Ph.D. degree in management science and engineering from Dalian University of Technology, Dalian, China, in 2006. \\
	From 2006 to 2008, he was a Post-Doctoral Research Fellow with Fudan University, Shanghai, China, where he is currently a Full Professor with the School of Computer Science. 
 Since 2019, he has been selected as one of the most cited Chinese researchers
	(Elsevier) every year. 
 His current research interests include network science, graph data mining, social network analysis, computational social science, spectral graph theory, and random walks. \\
	 Dr. Zhang was a recipient of the Excellent Doctoral Dissertation Award of Liaoning Province, China, in 2007, the Excellent Post-Doctor Award of Fudan University in 2008, the Shanghai Natural Science Award (third class) in 2013, the CCF Natural Science Award (second class) in 2022, and the Wilkes Award for the best paper published in The Computer Journal in 2019.
\end{IEEEbiography}

\end{document}